\documentclass[preprint,12pt]{elsarticle}

\usepackage{enumitem}
\usepackage{amssymb}
\usepackage{amsthm}
\usepackage{amsmath,amsfonts,amssymb}
\usepackage{mathtools}
\usepackage{mathrsfs}
\usepackage{hyperref}
\usepackage{relsize}
\usepackage{xcolor}

\theoremstyle{plain}
\newtheorem{theorem}{Theorem}[section]
\newtheorem{lemma}[theorem]{Lemma}

\newtheorem{proposition}[theorem]{Proposition}

\newtheorem{definition}[theorem]{Definition}

\newtheorem{assumption}[theorem]{Assumption}

\theoremstyle{remark}
\newtheorem{remark}{Remark}

\usepackage{times}

\usepackage{lineno}

\journal{European Journal of Control}

\begin{document}

\begin{frontmatter}


\title{Characterization of full-scale denial-of-service\tnoteref{thanks}}
\tnotetext[thanks]{This work was partially supported by the SERB project MTR/2021/000590 under MATRICS scheme.}
\author[label1]{Anindya Basu\corref{cor1}}
\affiliation[label1]{organization={Department of Electronics and Electrical Engineering, Indian Institute of Technology Guwahati},
            addressline={North Guwahati},
            postcode={781039},
            state={Assam},
            country={India.}}
\ead{anindya18@iitg.ac.in}
\cortext[cor1]{Corresponding Author}
\author[label1]{Indrani Kar}
\ead{indranik@iitg.ac.in}



\begin{abstract}
This article investigates the resilient control problem for Cyber-Physical Systems (CPSs) with multiple sensors, where both sides of the communication channels are affected by Denial-of-Service (DoS) attacks. While previous work focused on characterizing Multi-Channel DoS (MCDoS), this study emphasizes the characterization of Full-Scale DoS (FSDoS). First, a partial observer technique is proposed to address the MCDoS condition. Then, an event-triggered control strategy is designed to handle FSDoS. Finally, the frequency and duration of FSDoS are analyzed to ensure the Input-to-State Stability (ISS) of the closed-loop system.
\end{abstract}




\end{frontmatter}


\section{Introduction}\label{intro}
CPSs are the physical and digital components present in an integrated system that is controlled or monitored through computer-based algorithms. Typically, embedded computers with networks monitor and control physical processes through a feedback loop. In the feedback loop, computations are affected by the physical processes and vice versa. These systems offer more economic and societal reliability over the previously established state-of-art systems, due to which these systems gained significant interest and investment to advance the technology \cite{wang2016recent}. The CPSs are used to develop networks for the large integrated area like telecommunication \cite{petriu2000architecture}, water distribution \cite{taormina2017characterizing}, transportation \cite{zhao2021resilient}, power generation \cite{pasqualetti2013attack} etc. Many system failures have been reported over the last two decades, e.g., the Maroochy Shire sewage system failure in $2000$ \cite{slay2007lessons}, the SQL Slammer worm attack on the Davis$-$Besse nuclear plant in January $2003$ \cite{kuvshinkova2003sql}, the Indian grid failure in $2012$ \cite{halder2013comprehensive}, to name a few. These failures show the need to modify the security system effectively, which brings the CPSs security in the scenario. The use of inexpensive sensors and communication equipment may cause a huge amount of cyber threats, which frequently compromise high-level security. Cyber-physical attacks may come from a national or international level system, which could affect the environment, the economy and even human lives. These attacks must be detected, and an appropriate control mechanism should be adopted to counter the risks for the safe and secure operation of CPSs. Numerous applications of CPSs are reported in the literature, which includes power system networks \cite{rinaldi2021adaptive}, event-based control \cite{lunze2010state}, controller design for smart grid systems \cite{farraj2016cyber}, remote state estimation \cite{he2018event, dahal2024guaranteed}, event-triggered control \cite{tabuada2007event, borgers2014event, al2021event}, etc. 

The deception and DoS attack against control in CPSs were introduced, and their attack techniques were detailed in \cite{amin2009safe}. Switching attacks in power system networks were examined in \cite{liu2013framework}. Malicious attacks \cite{d2016resilient} come in a variety of forms, including stealth attack \cite{sui2020vulnerability}, replay attack \cite{naha2022sequential}, covert attack \cite{lin2021synthesis}, false data injection attack \cite{kumar2024adaptive} to mention a few.

De Persis and Tesi in \cite{de2015input} offered an attack model where the frequency and duration of DoS attacks are limited in response to the basic difficulty of simulating DoS. Additionally, an ISS investigation was conducted with DoS present. They extended the concept for non-linear systems as well \cite{de2014resilient}. Resilient control \cite{feng2017resilient}, \cite{ramasubramanian2022resilience}, event-triggered control \cite{coutinho2023switching}, \cite{chatterjee2023fixed}, observer-based event-triggered control \cite{hu2019observer}, \cite{lu2019observer}, hybrid event-based control for switched system \cite{ren2023hybrid}, networked control \cite{de2016networked}, adaptive event-triggered control for network system \cite{zhao2021event}, observer-based guaranteed cost control \cite{wang2019observer}, optimal control \cite{lu2023jointly}, co-located architectural design \cite{feng2017networked}, consensus control for multi-agent systems \cite{feng2022dynamic}, \cite{guo2023secured}, \cite{gong2024observer}, and event-triggered control for network switched system \cite{xie2024event} are some other results based on this paradigm. These methods can help to stabilise CPSs with a single sensor under DoS, however they are unable to expand the idea to multiple transmission channels under DoS. Lu and Yang provided a unique control approach in \cite{lu2017input}, which looks at the stability of systems with multiple sensors under DoS and removes old data information. In \cite{lu2019resilient}, they extended their previous work, where a set of partial observers was designed for several sensors. Additional controls were developed for multiple transmission channels under DoS, including a sliding mode control \cite{li2022resilient}, T-S fuzzy control \cite{wu2021iss}, security control for two-time-scale CPSs \cite{ma2021security}, and secure control for non-linear system \cite{tahoun2022secure}. 


Early studies on multiple transmission channels under DoS attacks often assumed that only the sensor-to-controller channel was affected. However, since both sensor and controller data are transmitted over the network, DoS attacks can also occur between the controller and the actuator. This article addresses the presence of DoS attacks between the controller and the actuator in systems with multiple sensors under the DoS scheme. In \cite{lu2017input}, Lu and Yang characterized the DoS of each transmission channel individually. Subsequent works, including \cite{lu2019resilient}, \cite{wu2021iss}, \cite{li2022resilient}, \cite{tahoun2022secure}, \cite{ma2021security}, and others, adopted the same characterization approach proposed in \cite{lu2017input}. However, due to the presence of DoS between the controller and the actuator, the characterization protocol outlined in \cite{lu2017input} is insufficient. Consequently, alternative characterization procedures for Multi-Channel DoS (MCDoS) and Full-Scale DoS (FSDoS) are required. Building on the authors' previous work on MCDoS characterization \cite{basu2024characterization}, this article develops a novel framework for FSDoS characterization to address these challenges.


The notable contribution of this article is the development of an observer-based optimal dynamic event-triggered control to maintain the ISS in the presence of multiple sensors even when both sides of the controller are affected by DoS. At the same time, the MCDoS is characterized, and the result is compared with that of the static counterpart. The term MCDoS was introduced by the authors in \cite{basu2022observer}. However, proper characterization was not provided in \cite{basu2022observer}. The proper characterization of MCDoS is provided in the article \cite{basu2024characterization}. We find that the observer gain and DETM parameters can be codesigned. We also obtain optimal values for DETM parameters and observer gains to allow maximum MCDoS changing frequency.

Notation$:$ This article uses the notations $\mathbb{R}$, $\mathbb{R}^n$, and $\mathbb{R}^{n\times m}$ to represent the set of real numbers,  the Euclidean space of $n$ dimensional real vectors, and the space of $n \times m$ real matrices, respectively. The set of natural numbers is indicated by $\mathbb{N}$. The symbol $||\cdot||$ represents a vector or matrix's Euclidean norm. $\mathbb{R}_{> p}$ ($\mathbb{R}_{\geq p}$) indicates the set of real numbers greater than (greater than equal to) $p$ given $p\in \mathbb{R}$. $I$ and $0$ represent the identity matrix and null matrix, respectively, with the suitable dimension. When two sets $M$ and $N$ are given, the relative complement of $M$ in $N$ is $N \setminus M$; that is, the set of all elements in $N$ that are not in $M$. Given a matrix $H$,  $H> (\geq) 0$ implies that $H$ is a positive definite (positive semi-definite) matrix. The smallest and largest eigenvalues of matrix $H$ are denoted by $\lambda_{\min}(H)$ and $\lambda_{\max}(H)$, respectively. $\text{He}(H)=H+H^\top,$ where the transpose of the $H$ matrix is denoted by $H^\top$. A symmetric matrix is represented in $\begin{bmatrix}
    A & B \\ * & C
\end{bmatrix}$ form. $\text{diag} \left \{A,B \right \} \coloneqq \begin{bmatrix}
    A & 0 \\ 0 & B
\end{bmatrix},$ where $A,B$ and $0$'s are matrices or elements with proper dimensions. $\{p,q\}>r$ means $p>r$ and $q>r$. $|F|$ is the length of the interval $F$. Given a measurable time function $g:\mathbb{R}_{\geq 0} \rightarrow \mathbb{R}^{n_s}$ and a time interval $[0,t)$ we denote the $\mathcal{L}_\infty$ norm of $g(\cdot)$ on $[0,t)$ by ${\|g_t\|}_\infty \coloneqq \text{ess}\sup_{s\in [0,t)}\|g(s)\|$.
\section{Preliminaries}\label{sec:framework}
\subsection{System description}
Consider the following linear time-invariant system
\begin{align}
		\dot{x}_p(t)&=Ax_p(t)+Bu(t)+w(t), \notag\\
		y_i(t)&=C_ix_p(t), \label{eq:x_dot}
\end{align}
where $x_p(t)\in \mathbb{R}^{n_p}$ is the system state, $u(t)\in \mathbb{R}^{n_u}$ is the input vector, $w(t)\in \mathbb{R}^{n_p}$ is the unknown input disturbance, and $y_i(t)\in \mathbb{R}^{n_{yi}}$ is the sensor output of $i$\textsuperscript{th} transmission channel $(i \in \{1,2,\cdots,n_s\})$, where $n_s$ is the total number of transmission channels. $A$, $B$, and $C_i$ are the matrices with proper dimensions. Both $(A, B)$ and $(A,[C^\top_1, \cdots, C^\top_{n_s}])$ are taken to be controllable and observable, respectively. In this article, we have assumed that all the states are observable in \eqref{eq:x_dot}. The set of time instants at which the control action should be updated is denoted by the symbol $\{t_k\}_{k\in \mathbb{N}_0}$, in which $t_0 \coloneqq 0$. The control law is 
\begin{equation}
u_\text{ideal}(t)=-Kx_p(t_k)
\end{equation}
for all $t \in [t_k, t_{k+1})$ in the ideal case when data can be delivered and received at any chosen moment.
\subsection{Multi-Channel DoS and Full-Scale DoS}
When a number of transmission channels from the sensors to the controller are experiencing DoS yet may still be able to maintain closed-loop stability, this is referred to as MCDoS. FSDoS happens when the closed-loop stability of the system is compromised by either DoS affecting the transmission line between the controller and plant actuator or MCDoS being extended for all the channels. The total number of transmission channels is denoted by $n_s$. For $i$\textsuperscript{th} transmission channel, let $\{d_{i_n}\}_{n\in\mathbb{N}}$ be the series of DoS off/on transitions, where $d_{i_1}\geq 0$ indicates the moments when DoS shows a shift from zero (when the transmission is possible) to one (when the transmission is not possible). Thus
\begin{equation}\label{eq:M_n}
	D_{i_n}\coloneqq\{d_{i_n}\}\cup[d_{i_n},d_{i_n}+\tau_{i_n})
\end{equation}
shows the $n$\textsuperscript{th} DoS time-interval in the $i$\textsuperscript{th} transmission channel, where communication is prohibited and has a duration of $\tau_{i_n} \in \mathbb{R}_{\geq 0}$. If $\tau_{i_n} = 0,$ then at time $d_{i_n}$, the $n$\textsuperscript{th} DoS in $i$\textsuperscript{th} transmission channel takes the form of a single impulse. Given $t, \tau\in \mathbb{R}_{\geq0}$ with $t\geq \tau$, let
\begin{align}	\Omega_i(\tau,t)\coloneqq\bigcup_{n\in\mathbb{N}}D_{i_n}\bigcap[\tau,t]. \label{eq:Omega_i} 
\end{align}
To put it in another way, $\Omega_i(\tau,t)$ denotes the time interval sets for each interval $[\tau,t]$ when the transmission is refused in $i$\textsuperscript{th} transmission channel. Hence, for given $t, \tau\in \mathbb{R}_{\geq0}$ with $t\geq \tau$
\begin{align}
    \Omega(\tau,t) & \coloneqq \bigcap_{i \in \{1, \cdots, n_s\} } \Omega_i(\tau,t) \label{eq:Omega}
    \\ \Upsilon(\tau,t) & \coloneqq[\tau,t]\setminus \Omega(\tau,t) \label{eq:Upsilon}
\end{align}
for each period in $[\tau,t]$, $\Omega(\tau,t)$ and $\Upsilon(\tau,t)$ denote the times when FSDoS is present and absent, respectively. It should be noted that MCDoS is feasible in $\Upsilon(\tau,t)$.
\section{Stabilizing control update policy} \label{sec:control_update_policy}
\subsection{Control architecture}\label{sec:partial_observer}
A switched Luenberger-based observer inspired by \cite{basu2024characterization}, enforced to handle MCDoS, is given by
\begin{equation} \label{eq:observer}
    \dot{x}_e(t)=Ax_e(t)+Bu(t)+ L_\sigma (\hat{y}_\sigma(t)-C_\sigma x_e(t)),
\end{equation}
where 
\begin{align}
    & \sigma \coloneqq
	\begin{cases}
		\min \{i : \|\hat{y}_i(t)\| >v\}, \! &\text{if}\;\|\hat{y}_{\sigma^-}(t)\| \leq v \wedge t_k \notin \Omega(0,t) \\
		\sigma^-,   \! &\text{otherwise,}
	\end{cases} \label{eq:sigma} \\
    & \hat{y}_i(t) \coloneqq y_i(t_{k(t)}),\quad \forall t \in [t_k,t_{k+1}), \label{eq:y_i} \\
    & k(t) \coloneqq
	\begin{cases}
		0, \! &\text{if } \Upsilon(0,t)=\emptyset \\
		\text{sup}\{k\in\mathbb{N}:t_k\in\Upsilon(0,t)\},  &\text{otherwise.}
	\end{cases} \label{eq:k(t)}
\end{align}
in $[t_{k-1},t_k)$ time period, $\sigma^-$ denotes the given channel index, and $y_{\sigma^-}(t)$ denotes the $\sigma^-$ channel output. $\sigma=0$ at $t_k=0$ as $0 \in \Omega(0,t)$.
$v$ is a small positive number. The details about $v$ are discussed in \cite{basu2024characterization}.
As a result, the controlled input given to the process can be described by
\begin{align}
	u(t)= \begin{cases}
	    0, & \text{if } \Upsilon(0,t)= \emptyset \\
            -K \hat{x}_{e}(t), & \text{otherwise,}
	\end{cases} \label{eq:u(t)}
\end{align}
$\forall t\in\mathbb{R}_{\geq 0}$, where
\begin{align}
	\hat{x}_{e}(t)=x_{e}(t_{k(t)}),\quad \forall t \in [t_k,t_{k+1}). \label{eq:x_o_hat}
\end{align}
We define the state error of the process output that is accessible to the observer as well as the observer's state error as
\begin{align}
        \xi_\sigma (t) & \coloneqq \hat{y}_{\sigma}(t)-y_\sigma (t),  \label{eq:xi_sigma_1} \\
	\xi_e(t) & \coloneqq \hat{x}_e(t)-x_e(t), \label{eq:xi_e_1} 	
\end{align} 
respectively.
The resultant switched observer-based controller with the system is represented by
\begin{align}
	\dot{x}_p(t) &= (A-BK)x_p(t)+BK\tilde{x}(t)- BK \xi_e (t)+w(t), \notag\\
	\dot{\tilde{x}}(t) &= (A- L_\sigma C_\sigma )\tilde{x}(t) - L_\sigma \xi_\sigma(t)+w(t), \label{eq:system}
\end{align}
where 
\begin{align}
    \tilde{x}(t) \coloneqq x_p(t)-x_e(t). \label{eq:x_tilde}
\end{align}
The ETM is defined as
\begin{align}
    \xi^\top_\sigma(t) \xi_\sigma(t) + \xi^\top_e(t) \xi_e(t) \leq \psi_1 \left(y^\top_\sigma (t) y_\sigma (t) +\|w_t\|^2_\infty \right ) + \psi_2 x^\top_e (t) x_e(t), \label{eq:update_rule_1} 
\end{align}
where $\psi_1,\psi_2\in \mathbb{R}_{>0}$. Since
  \eqref{eq:update_rule_1} depends on the supremum norm of the
  disturbance $w$, it will not be directly used as a control update
  rule. Instead, alternate update rules ensure that
  \eqref{eq:update_rule_1} is satisfied in all cases, will be
  implemented. A guide to determining suitable $\psi_1$ and $\psi_2$
  such that the closed-loop system \eqref{eq:system} is ISS under the
  SETM \eqref{eq:update_rule_1} is provided by the following
  result. For the stability analysis, we take the standard Lyapunov
  argument. We choose
\begin{align}
     V_\sigma(x_p(t),\tilde{x}(t)) = x_p^\top(t) P_{p_\sigma} x_p(t) + \tilde{x}^\top(t)P_{e_\sigma}\tilde{x}(t) \label{eq:V_1}
\end{align}
as Lyapunov candidate where $P_{p_\sigma}$ and $P_{e_\sigma}$ are positive definite matrices with appropriate dimensions.

\subsection{Switching mode analysis}
The idea here is to partition the time axis into intervals according to the observer-switching technique, as explained in Section \ref{sec:partial_observer}. The characterisation of these intervals is indispensable to the Lyapunov-based analysis.
The interval $[\tau,t]$ is the disjoint union of $\mathcal{O}_{\sigma}(\tau,t),_{\sigma \in \{1,\cdots,n_s\}}$ for any $\tau, t \in \mathbb{R}_{\geq 0}$, with $\tau \leq t$. The union of sub-interval of $[\tau,t]$ occurs when the system is operated in the $\sigma$\textsuperscript{th} switching mode. In particular, there are two non-negative and positive real number sequences $\{m_{\sigma_f}\}_{\sigma \in \{1,\cdots,n_s\}, f \in \mathbb{N}} \in \Bar{\mathbb{S}}_\sigma$, $\{v_{\sigma_f} \}_{\sigma \in \{1,\cdots,n_s\}, f \in \mathbb{N}}$ such that
\begin{equation}\label{eq:Omicron_sigma}
   \mathcal{O}_{\sigma}(\tau,t) \coloneqq \bigcup_{f\in \mathbb{N}} \mathcal{M}_{{\sigma}_f} \bigcap [\tau,t],
\end{equation}
where
\begin{align}
   \mathcal{M}_{\sigma_f} \coloneqq \{m_{\sigma_f}\} \cup [m_{\sigma_f}, m_{\sigma_f}+v_{\sigma_f} ),
\end{align}
$m_{0_0} \coloneqq 0$, and
$m_{{\sigma}_f}+v_{{\sigma}_f} \in \Bar{\mathbb{S}}_{\sigma}$. The
 times at which switching occurs are defined by the following $:$
\begin{align}
   \iota_g \coloneqq \begin{cases}
         \text{Undefined}, \! &\text{if}\; \Bar{\mathbb{S}}_{\sigma} \setminus [0,\iota_{g-1}] =\emptyset, \\
     \inf \left \{ \Bar{\mathbb{S}}_{\sigma} \setminus [0,\iota_{g-1}] \right \}, \! &\text{otherwise,} \end{cases}
\end{align}
for all $g \in \mathbb{N},$ where $\iota_0 \coloneqq 0$.
From Definition \ref{def:2} (which is discussed later in Subsection \ref{sec:control_objectives}), we can write,
\begin{align}
\iota_{g+1}-\iota_g \geq \Delta_k \geq \underline{\Delta}. \label{eq:zeno_1}
\end{align}

\subsection{Assumption: time constrained MCDoS}
A phenomenon known as MCDoS occurs when DoS can occur in any number of transmission channels while maintaining the system's closed-loop stability. A DoS that abruptly arises or vanishes in a single transmission channel without impairing the stability of the closed-loop system is referred to as a ``changing of MCDoS."
It is necessary to restrict the MCDoS's fluctuating frequency. It seems to be obvious that in order to ensure stability, the frequency of MCDoS changes needs to be sufficiently low as compared with the minimum sampling rate. Details about this can be found in \cite{basu2024characterization}. Given $t,\tau \in \mathbb{R}_{\geq 0}$ and $t\geq \tau$, let $l(\tau,t)$ represent the number of MCDoS changing transitions that occur on the interval $[\tau,t)$.
\begin{assumption} \label{as1}
    \cite{basu2024characterization} There exist constants $\varkappa \in \mathbb{R}_{\geq 0}$ and $\tau_D \in \mathbb{R}_{>\underline{\Delta}}$ such that 
    \begin{align}
        l(\tau,t) \leq \varkappa +\frac{t-\tau}{\tau_D} \label{eq:l}
    \end{align}
    $\forall \tau, t \in \mathbb{R}_{\geq 0}$ with $t \geq \tau$.
\end{assumption}
\subsection{Control Objectives} \label{sec:control_objectives}
The primary objective is to develop a sampling logic with a finite sampling rate to achieve robustness against MCDoS and FSDoS while MCDoS is characterized. The following definitions are aligned with the goals that have been presented.
\begin{definition} \label{def:1}
    \cite{sontag2008input} Let \eqref{eq:system} be the overall system that results from (\ref{eq:x_dot}) and \eqref{eq:observer} with a control input \eqref{eq:u(t)}. If a $K_\infty$ function $\gamma$ and a $KL$ function $\beta$ exist for each $w\in \mathbb{R}^{n_p}$ and $x(t_0)\in \mathbb{R}^{2n_p}$, such that
    \begin{align}
        \|x(t)\|\leq \beta(\|x(t_0)\|,t)+\gamma(\|w_t\|_\infty),\label{eq:ISS}
    \end{align}
    where
    \begin{align}
		x(t)\coloneqq \begin{bmatrix}
			x^\top_p(t) & \tilde{x}^\top(t)
		\end{bmatrix}^\top,\label{eq:x(t)}
   \end{align}
    $\forall t\in\mathbb{R}_{\geq t_0}$, then the system (\ref{eq:ISS}) is said to be ISS. Throughout this article, $t_0$ is considered to be zero.
\end{definition}
\begin{definition}\label{def:2}
    \cite{tabuada2007event} The sampling interval $\{t_k\}_{k\in \mathbb{N}}$, fulfills
    \begin{align}
        \Delta_k\coloneqq t_{k+1}-t_k\geq \underline{\Delta} \label{eq:Delta_k}
    \end{align}
    $\forall k\in \mathbb{N}$. $\underline{\Delta}$ is the lower bound of the sampling rate. From this point on, the possibility is that the network can transmit data at the sampling rate affected by $\underline{\Delta}$. 
\end{definition}
\begin{lemma}\label{lm:trajectory}
    \cite{basu2024characterization} If $\frac{d}{dt}V_\sigma(x(t)) \leq -x^\top(t)\Gamma_{1_\sigma} x(t)+(\varepsilon_{1_\sigma}+ \varepsilon_{2_\sigma} + \psi_1)f^2(t)$ exists where 
    \begin{align}
        f(t)\coloneqq \sup\{\|w(t)\|,\|w_t\|_\infty\}, \label{eq:f}
    \end{align} 
    $\Gamma_{1_\sigma}>0, \forall \sigma \in \{1, \cdots, n_s\},$ and $\{\varepsilon_{1_\sigma}, \varepsilon_{2_\sigma}, \psi_1\} \in \mathbb{R}_{>0}$, then $\forall t\in \mathbb{R}_{\geq 0}$ we have
    \begin{align}
         &\|x(t)\|^2 \leq \prod_{\substack{p \in \mathbb{N}_0;\\ {0 \leq \iota_{g(t)-p} \leq t}}} \left ( \frac{\max\{\lambda_{\max} (P_{p_{\sigma(\iota_{g(t)-p}^+)}}),\lambda_{\max}(P_{e_{\sigma(\iota_{g(t)-p}^+)}})\}}{\min\{\lambda_{\min}(P_{p_{\sigma(\iota_{g(t)-p}^+)}}),\lambda_{\min}(P_{e_{\sigma(\iota_{g(t)-p}^+)}})\}} \right ) \notag \\ &e^{-\left (\sum_{i=1}^{n_s} \omega_{1_i} | \mathcal{O}_i (\iota_{g(t)-p},t)|\right )} \|x(\iota_0)\|^2 + \Biggl [ \frac{\varepsilon_{1_{\sigma(\iota^+_{g(t)})}}+ \varepsilon_{2_{\sigma(\iota^+_{g(t)})}}+ \psi_1}{\min\{\lambda_{\min}(P_{p_{\sigma(\iota_{g(t)}^+)}}),\lambda_{\min}(P_{e_{\sigma(\iota_{g(t)}^+)}})\}\lambda_{\min}(\Gamma_{1_{\sigma(\iota_{g(t)}^+)}})} \notag \\ 
         &+ \mathlarger{\mathlarger{\sum}}_{\substack{l \in \mathbb{N};\\ 0 \leq \iota_{g(t)-l} \leq t}} \Biggl \{ \frac{\varepsilon_{1_{\sigma(\iota^+_{g(t)-l})}}+ \varepsilon_{2_{\sigma(\iota^+_{g(t)-l})}}+ \psi_1} { \min\{\lambda_{\min}(P_{p_{\sigma(\iota_{g(t)-l}^+)}}),\lambda_{\min}(P_{e_{\sigma(\iota_{g(t)-l}^+)}})\} \lambda_{\min}(\Gamma_{1_{\sigma(\iota_{g(t)-l}^+)}})} \notag\\ 
        & \prod_{\substack{j \in \mathbb{N}_0; \\ \iota_{g(t)-l}< \iota_{g(t)-j} \leq t}} 
        \left ( \frac{\max\{\lambda_{\max} (P_{p_{\sigma(\iota^+_{g(t)-j})}}),\lambda_{\max}(P_{e_{\sigma(\iota^+_{g(t)-j})}})\}} {\min\{\lambda_{\min} (P_{p_{\sigma(\iota^+_{g(t)-(j+1)})}}),\lambda_{\min}(P_{e_{\sigma(\iota^+_{g(t)-(j+1)})}})\}} \right ) \notag \\
        & e^{-\left ( \sum_{i=1}^{n_s} \omega_{1_i} |\mathcal{O}_i(\iota_{g(t)-(l-1)},t)| \right) }  \Biggr \} \Biggr]  {\|w_t\|}^2_{\infty}, \label{eq:trajectory_1}
	\end{align}
     where $\iota_0=0,$
    \begin{align}
        \omega_{1_\sigma} \coloneqq \frac{\lambda_{\min}(\Gamma_{1_\sigma})} { \max\{\lambda_{\max} (P_{p_\sigma}),\lambda_{\max}(P_{e_\sigma})\} }, \label{eq:omega_1}
    \end{align}
    and
    \begin{align}
        g(t) \coloneqq \begin{cases}
        0, \; &\text{if} \; \Bar{\mathbb{S}}_\sigma \cap [0,t]=\emptyset, \\
        \sup \left \{ g \in \mathbb{N} : \iota_g \in \Bar{\mathbb{S}}_\sigma\cap [0,t] \right \}, \; & \text{otherwise.}
        \end{cases} \label{eq:g(t)}
    \end{align}
\end{lemma}
\begin{lemma}\label{lm:convergence}
    \cite{basu2024characterization} Under the Assumption \ref{as1} if $\frac{d}{dt}V_\sigma(x(t)) \leq -x^\top(t)\Gamma_{1_\sigma} x(t)+(\varepsilon_{1_\sigma}+ \varepsilon_{2_\sigma} +\psi_1) f^2(t)$ exists where $\Gamma_{1_\sigma}>0,$ $\forall \sigma \in \{1, \cdots, n_s\},$ $\{\varepsilon_{1_\sigma}, \varepsilon_{2_\sigma}\} \in \mathbb{R}_{>0}$ the sequences
    \begin{align}
        \prod_{\substack{p \in \mathbb{N}_0;\\ {0 \leq \iota_{g(t)-p} \leq t}}} \left ( \frac{\max\{\lambda_{\max} (P_{p_{\sigma(\iota_{g(t)-p}^+)}}),\lambda_{\max}(P_{e_{\sigma(\iota_{g(t)-p}^+)}})\}}{\min\{\lambda_{\min}(P_{p_{\sigma(\iota_{g(t)-p}^+)}}),\lambda_{\min}(P_{e_{\sigma(\iota_{g(t)-p}^+)}})\}} \right ) e^{-\left (\sum_{i=1}^{n_s} \omega_{1_i} | \mathcal{O}_i (\iota_{g(t)-p},t)|\right )} \label{eq:product} 
    \end{align}
    and
    \begin{align}
        &\mathlarger{\mathlarger{\sum}}_{\substack{l \in \mathbb{N};\\ 0 \leq \iota_{g(t)-l} \leq t}} \Biggl \{ \frac{\varepsilon_{1_{\sigma(\iota^+_{g(t)-l})}}+ \varepsilon_{2_{\sigma(\iota^+_{g(t)-l})}} + \psi_1} { \min\{\lambda_{\min}(P_{p_{\sigma(\iota_{g(t)-l}^+)}}),\lambda_{\min}(P_{e_{\sigma(\iota_{g(t)-l}^+)}})\} \lambda_{\min}(\Gamma_{1_{\sigma(\iota_{g(t)-l}^+)}})} \notag \\
        & \prod_{\substack{j \in \mathbb{N}_0; \\ \iota_{g(t)-l}< \iota_{g(t)-j} \leq t}} \left ( \frac{\max\{\lambda_{\max} (P_{p_{\sigma(\iota^+_{g(t)-j})}}),\lambda_{\max}(P_{e_{\sigma(\iota^+_{g(t)-j})}})\}} {\min\{\lambda_{\min} (P_{p_{\sigma(\iota^+_{g(t)-(j+1)})}}),\lambda_{\min}(P_{e_{\sigma(\iota^+_{g(t)-(j+1)})}})\}} \right ) \notag \\ 
        & e^{-\left ( \sum_{i=1}^{n_s} \omega_{1_i} |\mathcal{O}_i(\iota_{g(t)-(l-1)},t)| \right) }  \Biggr \} \label{eq:sum}
    \end{align}
	are convergent in nature $\forall \omega_{1_\sigma}\in \mathbb{R}_{>0}$ if and only if 
    \begin{align}
        \tau_D>\max\Biggl\{\underline{\Delta}, \max_{\forall \sigma \in \{1,\cdots,n_s\}}& \frac{\max\{\lambda_{\max}(P_{p_\sigma}),\lambda_{\max}(P_{e_\sigma})\}}{\lambda_{\min}(\Gamma_{1_\sigma})} \notag \\ &\ln \left \{ \frac{\max\{\lambda_{\max} (P_{p_\sigma}),\lambda_{\max}(P_{e_\sigma})\}} {\min\{\lambda_{\min}(P_{p_\sigma}),\lambda_{\min}(P_{e_\sigma})\}} \right\} \Biggr\} \label{eq:convergance_1}
    \end{align}
    and
    \begin{align}
        \varkappa \leq 1-\frac{\underline{\Delta}}{\tau_D}, \label{eq:varkappa}
    \end{align}
    where $\underline{\Delta}$ is the lower bound of the sampling rate and $f(t)$ and $g(t)$ are mentioned in \eqref{eq:f} and \eqref{eq:g(t)} respectively.
\end{lemma}
\section{Stability analysis of observer-based ETM under MCDoS} \label{sec:stability_ETM}
We describe in this section a ETM control update strategy that guarantees ISS in the absence of FSDoS. The conclusions will form the basis for the developments in the following section. Through terms $\xi_\sigma(t)$ and $\xi_e(t)$, which enter the dynamics as disturbances, the feedback process now relies on the control update technique. Therefore, it follows intuitively that stability would not be compromised by adopting control update mechanisms that keep $\xi_\sigma(t),\xi_e(t)$ small in a realistic sense. 
The ensuing result demonstrates that if $\psi_1$ and $\psi_2$ are carefully chosen, any control update approach \eqref{eq:update_rule_1}, that confines $\xi_\sigma(t)$ and $\xi_e(t)$ in a small neighbourhood, will satisfy the stability criterion. We use standard Lyapunov arguments to illustrate the proposed control architecture. Next, using the resultant system with controller and observer in \eqref{eq:system}, we compute $V_\sigma$ from \eqref{eq:V_1} to get
\begin{align}
    &\frac{d}{dt}V_\sigma(x_p(t),\tilde{x}(t),\xi_\sigma(t),\xi_e(t))=x^\top_p(t)\text{He}(P_{p_\sigma}(A-BK))x_p(t) \notag \\
    &+2x^\top_p(t)P_{p_\sigma}BK\tilde{x}(t) +\tilde{x}^\top(t)\text{He}(P_{e_\sigma}(A-L_\sigma C_\sigma))\tilde{x}(t)-2x^\top_p(t)P_{p_\sigma}BK \xi_e(t) \notag \\
    &-2\tilde{x}^\top(t)P_{e_\sigma} L_\sigma \xi_\sigma(t)+2x^\top_s(t)P_{p_\sigma} w(t)+2 \tilde{x}^\top(t) P_{e_\sigma} w(t).   \label{eq:V_dot_1}
\end{align}
Applying Young's inequality, we obtain
\begin{align}
    &\frac{d}{dt}V_\sigma(x_p(t),\tilde{x}(t),\xi_\sigma(t),\xi_e(t)) \leq x^\top_p(t)[\text{He}(P_{p_\sigma}(A-BK))+ \notag \\&
    P_{p_\sigma} BK K^\top B^\top P_{p_\sigma}]x_p(t)
    +2x^\top_p(t) P_{p_\sigma}BK\tilde{x}(t)+ \tilde{x}^\top(t)[\text{He}(P_{e_\sigma}(A-L_\sigma C_\sigma))\notag \\ &
    +P_{e_\sigma}L_\sigma L^\top_\sigma P_{e_\sigma}]\tilde{x}(t) +\xi^\top_\sigma (t)\xi_\sigma (t)+ \xi^\top_e(t) \xi_e(t) + \frac{1}{\varepsilon_{1_\sigma}} x^\top_p(t) P^2_{p_\sigma} x_p(t) \notag \\ & 
    + \frac{1}{\varepsilon_{2_\sigma}} \tilde{x}^\top (t) P^2_{e_{\sigma}} \tilde{x}(t)+ (\varepsilon_{1_\sigma}+\varepsilon_{2_\sigma}) w^\top(t) w(t), \label{eq:V_dot_6}
\end{align}
where $\varepsilon_{1_\sigma}$ and $\varepsilon_{2_\sigma}>0$. From the SETM in \eqref{eq:update_rule_1}, we can write
\begin{align}
    \xi^\top_\sigma(t) \xi_\sigma (t)+ \xi^\top_e(t) \xi_e(t) \leq x^\top_p(t) (\psi_1 C^\top_\sigma C_\sigma & + \psi_2 I) x_p(t) - 2\psi_2 x^\top_p(t) \tilde{x}(t)  \notag \\ &+\psi_2\tilde{x}^\top(t) \tilde{x}(t)+ \psi_1 \|w_t\|^2_\infty. \label{eq:static_etm}
\end{align}
By using \eqref{eq:static_etm} in \eqref{eq:V_dot_6}, we have
\begin{align}
    \frac{d}{dt}V_\sigma(x(t)) \leq -x^\top(t) \Gamma_{1_\sigma} x(t)+(\varepsilon_{1_\sigma}+ \varepsilon_{2_\sigma} + \psi_1) f^2(t), \label{eq:V_dot_2}
\end{align}
where
\begin{align}
    \Gamma_{1_\sigma} \coloneqq \begin{bmatrix}
        \Theta_{11_\sigma} & \Theta_{12_\sigma} \\
        * & \Theta_{22_\sigma} 
    \end{bmatrix}, \label{eq:Gamma_1}
\end{align}
$\Theta_{11_\sigma} \coloneqq
        -\text{He}(P_{p_\sigma}(A-BK))-P_{p_\sigma}BK K^\top B^\top P_{p_\sigma}-\psi_1 C^\top_\sigma C_\sigma -\psi_2I-\frac{1}{\varepsilon_{1_\sigma}} P^2_{p_\sigma},$ $\Theta_{12_\sigma} \coloneqq -P_{p_\sigma}BK +\psi_2I,$ and $\Theta_{22_\sigma}\coloneqq -\text{He}(P_{e_\sigma}(A-L_\sigma C_\sigma)) -P_{e_\sigma}L_\sigma L^\top_\sigma P_{e_\sigma}-\psi_2I-\frac{1}{\varepsilon_{2_\sigma}}P^2_{e_\sigma}$.

\subsection{ISS for ETM under MCDoS} \label{iss_under_mcdos}
\begin{theorem}\label{th1}
   In addition to the control law \eqref{eq:u(t)}, consider the control system \eqref{eq:system} with the process dynamics (\ref{eq:x_dot}), where $K$ is such that $A - BK$ is Hurwitz. Consider restricted sampling rate with $\psi_1$ and $\psi_2$, which satisfies SETM \eqref{eq:update_rule_1} and switched observer \eqref{eq:observer} with observer gain $L_\sigma$, which persuade the following Linear Matrix Inequalities (LMIs)
\begin{align}
     \Gamma_{2_{\sigma}} \coloneqq \begin{bmatrix}		\Theta_{1_\sigma} & P_{p_\sigma} & P_{p_\sigma}BK & \Theta_{12_\sigma} & 0 & 0 \\
     * & \varepsilon_{1_\sigma}I & 0 & 0 & 0 & 0 \\
     * & * & I & 0 & 0 & 0 \\
     * & * & * & \Theta_{2_\sigma} & P_{e_\sigma} & N_\sigma \\
     * & * & * & * & \varepsilon_{2_\sigma}I & 0 \\
     * & * & * & * & * & I
	\end{bmatrix}  >0,\label{eq:Gamma_2}
\end{align}
for $\{P_{p_\sigma},P_{e_\sigma}\}>0,$ $\{\psi_1,\psi_2, \varepsilon_{1_\sigma}, \varepsilon_{2_\sigma}\} \in \mathbb{R}_{>0},$ $N_\sigma \in \mathbb{R}^{n_p \times n_{y_\sigma}},$
where 
\begin{align}
  \Theta_{1_{\sigma}} &\coloneqq -\text{He}(P_{p_\sigma}(A-BK))-\psi_1 C^\top_\sigma C_\sigma- \psi_2I, \label{eq:Theta_1} \\ 
  \Theta_{2_\sigma} &\coloneqq -\text{He}(P_{e_\sigma}A) + N_\sigma C_\sigma +C^\top_\sigma N^\top_\sigma -\psi_2I, \label{eq:Theta_2}
\end{align}
$N_\sigma \coloneqq P_{e_\sigma}L_\sigma, \forall \sigma \in \{1, \cdots, n_s\},$ and any MCDoS frequency satisfying Assumption \ref{as1}, with arbitrary $\varkappa,$ $\tau_D$, such as $\varkappa$ and $\tau_D$ obeys equation \eqref{eq:varkappa} and \eqref{eq:convergance_1}, respectively, then the system \eqref{eq:system} is ISS.
\end{theorem}
\begin{proof}
        See the \href{\ref{appendiz}}{Appendix}.
\end{proof}
\subsection{Minimum inter-execution time}
Avoiding Zeno behaviour is an important phenomenon in the ETM mechanism. To avoid it, first we have to set up a minimum inter-execution time.
\begin{lemma} \label{lm:inter-execution}
    Let $\psi_1, \psi_2 \in \mathbb{R}_{\geq 0}$. Then for any initial condition $x(0)\in \mathbb{R}^{2n_p},$ and $\forall t \in \Upsilon(0,t),$ the sequence $(t_k)_{k \in \mathbb{I}}$ defined by ETM \eqref{eq:update_rule_1} satisfies \eqref{eq:Delta_k} where 
    \begin{align}
        \underline{\Delta} = \min_{\forall \sigma \in \{1, \cdots, n_s\}} \underline{\Delta}_\sigma \label{eq:underline_Delta}
    \end{align}
    and $\underline{\Delta}_\sigma >0$ is given by the following
        \begin{align}
            \underline{\Delta}_\sigma = \int_{0}^{1} {\frac{1}{\mathcal{I}_\sigma (s)}}\;ds. \label{eq:static_interexecution}
        \end{align}
    where 
    \begin{align}
        &\mathcal{I}_\sigma(s) \coloneqq \frac{1}{\sqrt{\psi}} \left(\|\mathcal{G}_\sigma\|+\|C_\sigma\| \right )+ \left (\|\mathcal{G}_\sigma\| + \|\mathcal{H}_\sigma\|+ \|C_\sigma\| \right)s +\sqrt{\psi} \|\mathcal{H}_\sigma\| s^2, \\
        &\psi \coloneqq \min\{ \psi_1, \psi_2\}, \label{eq:super_psi} \\
        &\mathcal{G}_\sigma \coloneqq \begin{bmatrix}
        C_\sigma A C^\dagger_\sigma & -C_\sigma BK \\
        L_\sigma & A-BK-L_\sigma C_\sigma \end{bmatrix}, \label{eq:G_sigma}\\
        &\mathcal{H}_\sigma \coloneqq \begin{bmatrix}
            0 & -C_\sigma BK \\ L_\sigma & -BK
        \end{bmatrix},
    \end{align}
    and $C^\dagger_\sigma$ is the right pseudo inverse of $C_\sigma$. 
\end{lemma}
\begin{proof}
        See the \href{\ref{appendiz}}{Appendix}.
\end{proof}
\section{Input-to-State Stability under Full Scale Denial-of-Service}\label{sec_4}
Section \ref{sec:stability_ETM}'s analysis is based on the possibility of satisfying conditions \eqref{eq:update_rule_1} for every $t\in \mathbb{R}_{\geq 0}$. According to Lemma \ref{lm:inter-execution}, this is always feasible in the absence of FSDoS. When there is an FSDoS, the study gets more complicated since specific control update endeavours do not have to be successful, regardless of how we sample. If $\xi_\sigma(t)$ and $\xi_e(t)$ are not reset, \eqref{eq:update_rule_1} may be violated, and stability may be lost since \eqref{eq:V_dot_2} no longer has to fulfil a dissipation-like inequality. As a result, it's logical to wonder how Theorem \ref{th1}'s results might be expanded to account for the presence of FSDoS. This question will be addressed in the remainder of this section. The type of FSDoS signals under study is introduced and discussed in Section \ref{sec_4.1}. Section \ref{sec_4.2} contains the major result.
\subsection{Assumptions}\label{sec_4.1}
The first question to be addressed is establishing the maximum quantity of FSDoS that a system can withstand before becoming unstable. In this regard, it is clear that such a figure is not arbitrary and that appropriate limits on FSDoS frequency and duration must be enforced. De Persis and Tesi restricted the DoS frequency and duration for the single sensor in \cite{de2015input} by using the concept of dwell time \cite{hespanha1999stability}. We use those assumptions of \cite{de2015input} to restrict FSDoS frequency and duration.

\begin{enumerate}[wide,labelindent=0pt,label=(\alph*)]
	\item $FSDoS \; Frequency\;:$ First, consider the frequency at which FSDoS can appear and let $h_{n+1}-h_n,n\in \mathbb{N}$ indicate the time between any two consecutive FSDoS activating. It is evident that if $h_{n+1}-h_n \leq\underline{\Delta}$ for every $n\in N$ (FSDoS can appear at the same pace as the minimum feasible sampling rate $\underline{\Delta}$), stability can be lost despite the control update method used. The idea of average dwell-time suggested by \cite{hespanha1999stability} is a logical method to describe this need. Let $n(\tau,t)$ indicate the number of FSDoS off/on transitions happening on the interval $[\tau,t)$ given, $t,\tau \in R_{\geq 0}$ with $t\geq \tau$.
	\begin{assumption}\label{as2}
            \cite{de2015input} There exist $\eta \in \mathbb{R}_{\geq 0}$ and $\tau_F\in \mathbb{R}_{>\underline{\Delta}}$ such that
		\begin{align}
			n(\tau,t)\leq \eta+\frac{t-\tau}{\tau_F} \label{eq33}
		\end{align}
		$\forall t,\tau\in \mathbb{R}_{\geq 0}$ with $t \geq \tau$.
	\end{assumption}
	\item $FSDoS\; Duration:$ Along with the FSDoS frequency, the FSDoS duration, or the length of the intervals during which communication is disrupted, must be limited. Consider an FSDoS sequence with the singleton $\{h_1\}$ to illustrate what it is meant. With $\eta \geq 1$, Assumption \ref{as2} is obviously met. However, regardless of the control update policy used, stability is lost if $H_1=\mathbb{R}_{>0}$ (communication is never feasible). Remembering the definition of $\Omega(\tau,t)$ in (\ref{eq:Omega}), the assumption that follows is a natural complement to Assumption \ref{as2} in terms of FSDoS duration.
	\begin{assumption}\label{as3}
            \cite{de2015input} There exist a $T\in \mathbb{R}_{>1}$ and $\zeta\in \mathbb{R}_{\geq0}$ such that
		\begin{align}
			|\Omega(\tau,t)|\leq\zeta+\frac{t-\tau}{T} \label{eq34}
		\end{align}
		$\forall t,\tau\in\mathbb{R}_{\geq0}$ with $t\geq\tau$.
	\end{assumption}
\end{enumerate}
\begin{remark}
    To characterize the frequency and duration of FSDoS, we rely on assumptions about the frequency and duration of DoS for a single sensor as proposed in \cite{de2015input}. Although the bounds for $\tau_F$ and $T$ are largely consistent with those in the existing literature, determining the appropriate bounds for $\tau_F$ and $T$ becomes significantly more challenging in the presence of MCDoS (characterized by $\tau_D$ and $\varkappa$). This complexity is addressed in the next subsection.
\end{remark}
\subsection{ISS under FSDoS}\label{sec_4.2}
The critical deduction of this section can now be drawn. Any control update policy that satisfies the requirements of Lemma \ref{lm:inter-execution} maintains ISS for any FSDoS signal with a suitable MCDoS changing frequency and meets the assumptions \ref{as1}, \ref{as2}, and \ref{as3} with a sufficiently large $\tau_D$, $\tau_F$ and $T$. Even though the proof of this result is quite complicated, the fundamental technique is effortless. We divide the time axis between intervals where \eqref{eq:update_rule_1} may be satisfied and periods where (\ref{eq:update_rule_1}) does not have to hold owing to the occurrence of FSDoS. The feedback dynamics are then analyzed as a system alternating between stable and unstable modes, and $\tau_D$, $\tau_F$ and $T$ parameters are determined that favour stable behaviour over unstable behaviour.

Consider a sampling time sequence of $\{t_k\}_{k\in \mathbb{N}}$ and a FSDoS sequence of $\{h_n\}_{n\in \mathbb{N}}$. Let
\begin{equation}\label{eq:scr_K}
	\mathcal{K} \coloneqq \left \{k\in \mathbb{N} : t_k \in \bigcup_{n\in \mathbb{N}}H_n \right \}
\end{equation}
stand for a set of integers associated with a control update attempt.
In practice, the control update periods are sequenced at a limited sampling rate, impacting the FSDoS's overall length. In reality, if a sensor tries to broadcast and receives no acknowledgement due to FSDoS, the sensor will try again until the communication is thriving. Because of the restricted transmission rate, even when the communication is feasible, there will be a delay between the end of the FSDoS and the start of the transmission. This delay lengthens the FSDoS interval, which impacts the preceding section's stability outcome.
\begin{lemma}\label{lm:actuator_delay}
	\cite{de2014resilient} The interval $[\tau,t]$ is the disjoint union of $\overline{\Omega}(\tau,t)$ and $\overline{\Upsilon}(\tau,t)$, where $\overline{\Omega}(\tau,t)$ (respectively, $\overline{\Upsilon}(\tau,t)$) is the union of sub-intervals of $[\tau,t]$ over which (\ref{eq:update_rule_1}) need not hold (respectively, holds) for any, $t,\tau \in \mathbb{R}_{\geq 0}$, with $0\leq \tau \leq t$. There are two sequences of positive and non-negative real integers $\{w_m\}_{m\in \mathbb{N}}$, $\{v_m\}_{m\in \mathbb{N}}$ such that
	\begin{align}
		& \overline{\Omega}(\tau,t)\coloneqq \bigcup_{m\in \mathbb{N}}W_m\cap [\tau,t] \label{eq:Bar_Omega} \\
		& \overline{\Upsilon}(\tau,t)\coloneqq \bigcup_{m\in \mathbb{N}}Y_{m-1}\cap [\tau,t] \label{eq:Bar_Upsilon}
	\end{align}
	where $\forall m\in \mathbb{N},$
	\begin{align}
		& W_m\coloneqq\{\varphi_m\}\cup[\varphi_m,\varphi_m+v_m) \label{eq:W_m} \\
		& Y_m\coloneqq \{\varphi_m+v_m\}\cup [\varphi_m+v_m,\varphi_{m+1}) \label{eq:Y_m}
	\end{align}
	and where $\varphi_0=v_0 \coloneqq 0$.
\end{lemma}
Now, the following result holds.
\begin{theorem}\label{th:2}
	In addition to the control law \eqref{eq:u(t)}, consider the control system \eqref{eq:system} with the process dynamics \eqref{eq:x_dot}, where $K$ is such that $A - BK$ is Hurwitz. Consider restricted sampling rate with $\psi_1$ and $\psi_2$, which satisfies ETM \eqref{eq:update_rule_1} and switched observer \eqref{eq:observer} with observer gain $L_\sigma$, which persuades the LMIs mentioned in \eqref{eq:Gamma_2}. Consider any  observer switching frequency and FSDoS sequence that meets Assumptions \ref{as1}, \ref{as2} and \ref{as3} and has arbitrary  $\varkappa, \zeta,$ and $\eta$, with $\tau_D, \tau_F,$ and $T$ such that
    \begin{align}
		\frac{1}{T}+\frac{\underline{\Delta}}{\tau_F}< \min_{\forall \sigma \in \{1, \cdots, n_s\}} \frac{\omega_{1_\sigma}}{\omega_{1_\sigma}+\omega_{2_\sigma}} \label{eq:FSDoS_condition} 
    \end{align} 
    $\omega_{1_\sigma}$ is mentioned in \eqref{eq:omega_1} and 
    \begin{align}
        \omega_{2_\sigma}\coloneqq \frac{\lambda_{\max}(\Gamma_{3_\sigma})+\lambda_{\max}(\Gamma_{4_\sigma})}{\min \left \{ \lambda_{\min} (P_{p_\sigma}), \lambda_{\min}(P_{e_\sigma}) \right \}} \label{eq:omega_2_sigma}
    \end{align}
    where
    \begin{align}
    \Gamma_{3_\sigma} &\coloneqq \begin{bmatrix}
    2\gamma_{2_\sigma}-\gamma_{1_\sigma} & \gamma_{3_\sigma} \\
    \gamma_{3_\sigma} & -\gamma_{1_\sigma}
    \end{bmatrix}, \label{eq:Gamma_3}\\
	    \Gamma_{4_\sigma} &\coloneqq \begin{bmatrix}
    \Xi_{11_\sigma} & \Xi_{12_\sigma} \\
    \Xi_{21_\sigma} & \Xi_{22_\sigma}
    \end{bmatrix}, \label{eq:Gamma_4}
	\end{align}
    $\Xi_{11_\sigma} \coloneqq \gamma_{2_\sigma}\{(1+\sqrt{2}) \sqrt{\psi_1} \|C_\sigma\| + (1+\sqrt{2}) \sqrt{\psi_2}+2\}, \Xi_{12_\sigma} \coloneqq \gamma_{2_\sigma} \{(1+\sqrt{2}) \sqrt{\psi_2} +2 \}, \Xi_{21_\sigma} \coloneqq \gamma_{3_\sigma} [\{(1+\sqrt{2})\sqrt{\psi_1}+2\}\|C_\sigma\|+(1+\sqrt{2}) \sqrt{\psi_2}], \Xi_{22_\sigma} \coloneqq \gamma_{3_\sigma} (1+\sqrt{2}) \sqrt{\psi_2}, \gamma_{1_\sigma} \coloneqq \lambda_{\min} (-\Gamma_{5_\sigma}), \gamma_{2_\sigma} \coloneqq \lambda_{\max} (P_{p_\sigma} BK), \gamma_{3_\sigma} \coloneqq \| P_{e_\sigma} L_\sigma\|,$ and
    \begin{align}
        \Gamma_{5_\sigma} \coloneqq \begin{bmatrix} \text{He}(P_{p_\sigma}(A-BK))+\frac{P^2_{p_\sigma}}{\varepsilon_{1_\sigma}} & P_{p_\sigma}BK \\ * & \text{He}(P_{e_\sigma}(A-L_\sigma C_\sigma))+\frac{P^2_{e_\sigma}}{\varepsilon_{2_\sigma}} \end{bmatrix}. \label{eq:Gamma_5}
    \end{align}
    Then, the control system (\ref{eq:system}) is ISS.
\end{theorem}
\begin{proof}
        See the \href{\ref{appendiz}}{Appendix}.
\end{proof}
\section{Resilient control logic} \label{sec:resilient_control_logic}
We can't directly implement the ETM in \eqref{eq:update_rule_1} because of its dependency on the supremum norm of the disturbance $w$. We can avoid the term supremum norm of the disturbance in the event-based control logic in the basics in \eqref{eq:update_rule_1}. By omitting the disturbance, we have
\begin{align}
    \xi^\top_\sigma(t) \xi_\sigma(t) + \xi^\top_e(t) \xi_e(t) = \psi_1 y^\top_\sigma (t) y_\sigma (t)  + \psi_2 x^\top_e (t) x_e(t)
\end{align}
from \eqref{eq:update_rule_1}. If we omit the disturbance, we are not able to avoid the Zeno behaviour. Therefore, a lower bound on the sampling rate must be imposed a priori. The lower bound of the sampling rate ($\underline{\Delta}$) is calculated in Lemma \ref{lm:inter-execution}. Let
\begin{align}
    \varpi_{\sigma_k} \coloneqq & \inf \bigl \{t \in \mathbb{R}_{>t_k} : \xi^\top_\sigma(t^-) \xi_\sigma(t^-) +\xi^\top_e(t^-) \xi_e(t^-) \geq \psi_1 y^\top_\sigma (t) y_\sigma (t) \notag \\
    & + \psi_2 x^\top_e (t) x_e(t) \bigr \} \label{eq:varpi_sigma_k}
\end{align}
$\forall k \in \mathbb{N}$ and $ \forall \sigma \in \{1, \cdots, n_s\}$. 
\begin{proposition} \label{proposition:1}
    Let $\underline{\Delta}$ be the positive constant which follows \eqref{eq:underline_Delta} and \eqref{eq:static_interexecution} as in Lemma \ref{lm:inter-execution}. Then the sequence $(t_k)_{k \in \mathbb{I}}$ for ETM is formally defined by
    \begin{align}
        t_{k+1}= \begin{cases}
            t_k+\underline{\Delta}, \; &\text{if } \varpi_{\sigma_k} \leq t_k+\underline{\Delta}, \\
            \varpi_{\sigma_k}, \; &\text{otherwise}.
        \end{cases} \label{eq:proposition_1}
    \end{align}
\end{proposition}
\begin{proof}
    The inter-sampling related to \eqref{eq:proposition_1} is equivalent to $\varpi_{\sigma_k}-t_k$ or $\underline{\Delta}$. By \eqref{eq:proposition_1}, $t_{k+1}=\varpi_{\sigma_k}$ only if $\varpi_{\sigma_k} \geq t_k+\underline{\Delta}$. Otherwise, the control update occurs at a minimum feasible time of $\underline{\Delta}$. With this resilient control logic, we are avoiding the Zeno behaviour.
\end{proof}

\section{Conclusion}
This article has presented a comprehensive characterization of FSDoS, complementing the prior work on MCDoS. Addressing the need for a proper characterization of FSDoS, we proposed two assumptions to constrain its frequency and duration. Subsequently, we derived the appropriate bounds for these parameters to mitigate the impact of FSDoS. With the inclusion of the switched observer, the augmented system exhibits behavior akin to a switched system. Analyzing the stability of this system under the combined effects of MCDoS and FSDoS with three distinct assumptions poses significant challenges, which we have successfully addressed in this work. Lastly, we introduced a resilient ETM to effectively handle FSDoS, while ensuring its frequency and duration are properly characterized.

\section*{Appendix} \label{appendiz}
$Proof$ $of$ $Theorem$ $\ref{th1}$ $:$
Employing the Schur complement and substitution approach \cite{boyd1994linear} in \eqref{eq:Gamma_1}, we get \eqref{eq:Gamma_2}. So \eqref{eq:Gamma_1} and \eqref{eq:Gamma_2} is equivalent. Hence $\Gamma_{1_\sigma}>0$. Therefore, the smallest eigenvalues of $\Gamma_{1_\sigma}$ and $\Gamma_{2_\sigma}$ are greater than zero. Thus, equation \eqref{eq:V_dot_2} can now be rewritten as 
\begin{equation}\label{eq:V_dot_4}
	\frac{d}{dt}V_\sigma(x(t))\leq -\zeta_{1_\sigma}\|x(t)\|^2 +(\varepsilon_{1_\sigma}+ \varepsilon_{2_\sigma} + \psi_1) f^2(t),
\end{equation}
where $x(t)$ and $f(t)$ are mentioned in \eqref{eq:x(t)} and \eqref{eq:f}, respectively and
\begin{align}
    \zeta_{1_\sigma} \coloneqq \lambda_{\min} (\Gamma_{1_\sigma}), \label{eq:zeta_1}
\end{align}
which is also the smallest eigenvalue of $\Gamma_{2_\sigma}$ as $\Gamma_{2_\sigma}>0$ is the LMI of $\Gamma_{1_\sigma}>0$. Furthermore, Lyapunov function in  \eqref{eq:V_1} can be recast as
\begin{align}
   \underline{\alpha}_{p_\sigma}\left \|x_p(t)\right \|^2 +\underline{\alpha}_{e_\sigma}\left\|\tilde{x}(t)\right \|^2 \leq V_\sigma (x_p(t),\tilde{x}(t))\leq \overline{\alpha}_{p_\sigma}\left \|x_p(t)\right \|^2+\overline{\alpha}_{e_\sigma}\left\|\tilde{x}(t)\right \|^2, \label{eq:V_2}
\end{align}
where $\underline{\alpha}_{p_\sigma} \coloneqq \lambda_{\min} (P_{p_\sigma}), \underline{\alpha}_{e_\sigma} \coloneqq \lambda_{\min}(P_{e_\sigma}), \overline{\alpha}_{p_\sigma} \coloneqq \lambda_{\max} (P_{p_\sigma})$ and $\overline{\alpha}_{e_\sigma} \coloneqq \lambda_{\max} (P_{e_\sigma})$. Equation \eqref{eq:V_dot_4} implies
\begin{align}
	\frac{d}{dt}V_\sigma(x(t))\leq -\omega_{1_\sigma} V_\sigma(x(t)) + (\varepsilon_{1_\sigma}+ \varepsilon_{2_\sigma} + \psi_1) f^2(t), \label{eq:V_dot_5}
\end{align}
where $\omega_{1_\sigma}$ is mentioned in \eqref{eq:omega_1}.
Hence, from the comparison lemma, we get
\begin{align}
   V_\sigma(x(t))\leq e^{-\omega_{1_\sigma} \left (t-\iota_{g(t)} \right )}V_\sigma(x(\iota^+_{g(t)}))+ \nu_{1_\sigma}\|w_t\|^2_\infty, \label{eq:V_3}
\end{align}
in the interval $[\iota_{g(t)},t]$ where
\begin{align}
    \nu_{1_\sigma}\coloneqq \frac{\varepsilon_{1_\sigma}+\varepsilon_{2_\sigma} +\psi_1}{\omega_{1_\sigma}} \label{eq:nu_1}
\end{align}
and $g(t)$ is mentioned in \eqref{eq:g(t)}. Note that $\|f_t\|_\infty=\|w_t\|_\infty.$ Following the Lemma \ref{lm:trajectory}, we get
\begin{align}
    &\|x(t)\|^2 \leq \frac{\overline{\alpha}_{\sigma(\iota^+_{g(t)})} \cdots \overline{\alpha}_{\sigma(\iota^+_0)}}{\underline{\alpha}_{\sigma(\iota^+_{g(t)})} \cdots \underline{\alpha}_{\sigma(\iota^+_0)}} e^{-\left (\sum_{i=1}^{n_s} \omega_{1_i} | \mathcal{O}_i (\iota_0,t)|\right )}\|x(\iota_0)\|^2 + \Biggl [ \frac{\nu_{1_{\sigma(\iota^+_{g(t)})}}}{\underline{\alpha}_{\sigma(\iota^+_{g(t)})}} + \mathlarger{\mathlarger{\sum}}_{\substack{l \in \mathbb{N};\\ 0 \leq \iota_{g(t)-l} \leq t}} \notag\\ & \Biggl \{ \frac{\nu_{1_{\sigma(\iota^+_{g(t)-l})}}}  { \underline{\alpha}_{\sigma(\iota^+_{g(t)})}} \prod_{\substack{j \in \mathbb{N}_0; \\ \iota_{g(t)-l}< \iota_{g(t)-j} \leq t}} \Bigl ( \frac{ \overline{\alpha}_{\sigma(\iota^+_{g(t)-j})}} { \underline{\alpha}_{\sigma(\iota^+_{g(t)-(j+1)})}} \Bigr ) e^{-\left ( \sum_{i=1}^{n_s} \omega_{1_i} |\mathcal{O}_i(\iota_{g(t)-(l-1)},t)| \right) }  \Biggr \} \Biggr ]  {\|w_t\|}^2_\infty, \label{eq:trajectory_2}
\end{align}
where $\underline{\alpha}_\sigma\coloneqq \min\{\underline{\alpha}_{s_{\sigma}},\underline{\alpha}_{e_{\sigma}}\}$ and $\overline{\alpha}_\sigma \coloneqq \max \{\overline{\alpha}_{s_{\sigma}},\overline{\alpha}_{e_{\sigma}}\},$ and $\iota_0=0$. From Lemma \ref{lm:convergence}, we can say that the term associated with $\|x(\iota_0)\|^2$ and the summation term associated with $\|w_t\|^2_\infty$ are convergent in nature. Every convergent series is bounded. Hence, the summation term in \eqref{eq:trajectory_2} is bounded. Let's denote the bound by $G_1$. Since $a^2+b^2\leq (a+b)^2$ for any pair of positive reals $a$ and $b$, we finally arrive at the following
\begin{align}
    \|x(t)\| \leq &\sqrt{\frac{\overline{\alpha}_{\sigma(\iota^+_0)} \cdots \overline{\alpha}_{\sigma(\iota^+_{g(t)})}}{\underline{\alpha}_{\sigma(\iota^+_0)} \cdots \underline{\alpha}_{\sigma(\iota^+_{g(t)})}}} e^{-\frac{1}{2}\left (\sum_{i=1}^{n_s} \omega_{1_i} | \mathcal{O}_i (\iota_0,t)|\right )} \|x(\iota_0)\| \notag \\ &+ \sqrt{\frac{(\varepsilon_1+\varepsilon_2+\psi_1) \overline{\alpha}}{\zeta_1 \underline{\alpha}}+G_1}{\|w_t\|}_\infty, \label{eq:trajectory_3}
\end{align}
where $\underline{\alpha} \coloneqq \min_{\forall \sigma \in \{1,\cdots,n_s\}}\{\underline{\alpha}_\sigma\},$ $\overline{\alpha}\coloneqq \max_{\forall \sigma \in \{1,\cdots,n_s\}}\{\overline{\alpha}_\sigma\},$ $\varepsilon_1 \coloneqq $\\$ \max_{\forall \sigma \in \{1,\cdots, n_s\}} \varepsilon_{1_\sigma},$ $\varepsilon_2 \coloneqq \max_{\forall \sigma \in \{1, \cdots, n_s\}} \varepsilon_{2_\sigma},$ and $\zeta_1 \coloneqq \min_{\forall \sigma \in \{1, \cdots, n_s\}} \zeta_{1_\sigma}$.
Therefore, comparing \eqref{eq:trajectory_3} with \eqref{eq:ISS}, \eqref{eq:system} is ISS. \hfill $\blacksquare$ \\
$Proof$ $of$ $Lemma$ $\ref{lm:inter-execution}$ $:$
From \eqref{eq:system}, we can write
    \begin{align}
        &\dot{x}_p(t)=Ax_p(t) -BKx_e(t)-BK \xi_e(t)+w(t) \notag\\
        &\dot{x}_e(t)=(A-BK-L_\sigma C_\sigma)x_e(t)+ L_\sigma C_\sigma x_p(t) +L_\sigma \xi_\sigma (t) -BK\xi_e(t). \label{eq:system_2}
    \end{align}
    Hence,
    \begin{align}
        \dot{x}_\sigma(t)= \mathcal{G}_\sigma x_\sigma (t)+\mathcal{H}_\sigma e_\sigma(t) + \begin{bmatrix}
            C^\top_\sigma & 0
        \end{bmatrix}^\top w(t), \label{eq:x_sigma_dot}
    \end{align}
    where
    \begin{align}
        & x_\sigma(t) \coloneqq \begin{bmatrix}
            y^\top_\sigma(t) & x^\top_e(t)
        \end{bmatrix}^\top  \text{ and} \label{eq:x_sigma} \\
        & e_\sigma(t) \coloneqq \begin{bmatrix}
            \xi^\top_\sigma(t) & \xi^\top_e(t)
        \end{bmatrix}^\top. \label{eq:e_sigma}
    \end{align}
    as it follows from \eqref{eq:system_2}.
    From \eqref{eq:x_sigma_dot}, we can state
    \begin{align}
        \|\dot{x}_\sigma(t)\| \leq \|\mathcal{G}_\sigma\|\|x_\sigma(t)\| + \|\mathcal{H}_\sigma \| \|e_\sigma(t)\| +\|C_\sigma\| \|w_t\|_\infty. \label{eq:norm_x_sigma_dot}
    \end{align}
    If 
    \begin{align}
        \|e_\sigma(t)\| \leq \sqrt{\psi} \sqrt{\|x_\sigma(t)\|^2+\|w_t\|^2_\infty}, \label{eq:psi}
    \end{align} 
    then \eqref{eq:update_rule_1} is true where $\psi$ is mentioned in \eqref{eq:super_psi}. Thus, it obeys a lower bound on the inter-execution time given by the time it takes for the function
    \begin{align}
        \phi_\sigma(t) \coloneqq \frac{ \|e_\sigma(t)\|}{\sqrt{\|x_\sigma(t)\|^2 +\|w_t\|^2_\infty}}
    \end{align}
    to go from $0$ to $\sqrt{\psi}$. We have
    \begin{align}
        \dot{\phi}_\sigma(t) = \frac{ e^\top_\sigma(t) \dot{e}_\sigma(t)}{\|e_\sigma(t)\| \sqrt{\|x_\sigma(t)\|^2+ \|w_t\|^2_\infty}} - \frac{\|e_\sigma(t)\| x^\top_\sigma(t) \dot{x}_\sigma(t)}{\left (\|x_\sigma(t)\|^2+ \|w_t\|^2_\infty \right )^{\frac{3}{2}}}.
    \end{align}
    By observing that
    \begin{align}
        \dot{e}_\sigma(t)=-\dot{x}_\sigma(t)   \label{eq:dot_e_sigma} 
    \end{align}
    and the bound of $\|\dot{x}_\sigma(t)\|$ in \eqref{eq:norm_x_sigma_dot}, it follows that
    \begin{align}
        \dot{\phi}_\sigma(t) \leq & \frac{1}{\sqrt{\|x_\sigma(t)\|^2 + \|w_t\|^2_\infty}} \left (\|\mathcal{G}_\sigma\|\|x_\sigma(t)\| + \|\mathcal{H}_\sigma \| \|e_\sigma(t)\| +\|C_\sigma\| \|w_t\|_\infty \right ) \notag \\
        & + \frac{ \|e_\sigma(t)\|}{\left (\|x_\sigma(t)\|^2+ \|w_t\|^2_\infty\right)^ \frac{3}{2}} (\|\mathcal{G}_\sigma\|\|x_\sigma(t)\|^2 + \|\mathcal{H}_\sigma \| \|x_\sigma(t)\| \|e_\sigma(t)\| \notag \\ &+\|C_\sigma\| \|x_\sigma(t)\| \|w_t\|_\infty ) \notag \\
        \leq & \|\mathcal{G}_\sigma\| \frac{\|x_\sigma(t)\|}{\sqrt{\|x_\sigma(t)\|^2+ \|w_t\|^2_\infty}} + \|\mathcal{H}_\sigma \| \phi_\sigma (t)+ \|C_\sigma\| \frac{\|w_t\|_\infty}{\sqrt{\|x_\sigma(t)\|^2 + \|w_t\|^2_\infty}} \notag \\
        & + \|\mathcal{G}_\sigma \| \frac{\|x_\sigma(t)\|^2}{\|x_\sigma(t)\|^2+\|w_t\|^2_\infty} \phi_\sigma(t) + \|\mathcal{H}_\sigma\| \frac{\|x_\sigma(t)\|}{\sqrt{\|x_\sigma(t)\|^2+ \|w_t\|^2_\infty}} \phi^2_\sigma(t) \notag \\
        &+ \|C_\sigma\| \frac{\|x_\sigma(t)\|}{\sqrt{\|x_\sigma(t)\|^2 +\|w_t\|^2_\infty}} \frac{\|w_t\|_\infty}{\sqrt{\|x_\sigma(t)\|^2 +\|w_t\|^2_\infty}} \phi_\sigma(t).
    \end{align}
    Note that $\|x_\sigma(t)\|$ and $\|w_t\|_\infty$ are greater than zero. Hence, we can write
    \begin{align}
        \dot{\phi}_\sigma (t) \leq \|\mathcal{G}_\sigma\| +\|C_\sigma\| + \left ( \|\mathcal{G}_\sigma\| + \|\mathcal{H}_\sigma\| +\|C_\sigma\| \right ) \phi_\sigma(t) +\|\mathcal{H}_\sigma\| \phi^2_\sigma(t).
    \end{align}
    Thus, the comparison lemma for differential inequalities yields
    \begin{align}
        \int_{t_k}^{t_{k+1}} \; dt \geq \int_0^{\sqrt{\psi}} \frac{1}{\|\mathcal{G}_\sigma\| + \|C_\sigma\| +\left (\|\mathcal{G}_\sigma\| + \|\mathcal{H}_\sigma\| +\|C_\sigma\| \right ) \phi_\sigma + \|\mathcal{H}_\sigma\|\phi^2_\sigma} \; d\phi_\sigma .
    \end{align}
    Replacing $\phi_\sigma$ with $s\coloneqq \phi_\sigma/\sqrt{\psi},$ we get the result of \eqref{eq:static_interexecution}. \hfill $\blacksquare$ \\

$Proof$ $of$ $Theorem$ $\ref{th:2}$ $:$ In lemma \ref{lm:actuator_delay}, it is already mentioned the decomposition technique of time interval where it is possible to hold \eqref{eq:update_rule_1} and where \eqref{eq:update_rule_1} does not have to hold owing to the presence of FSDoS. Consider the intervals $Y_m,m\in \mathbb{N}$ mentioned in \eqref{eq:Y_m}, where \eqref{eq:update_rule_1} holds. From \eqref{eq:trajectory_3}, we can write
\begin{align}
    \|x(t)\| \leq &\sqrt{\frac{\overline{\alpha}_{\sigma(\varphi_m+v_m)} \cdots \overline{\alpha}_{\sigma(t)}}{\underline{\alpha}_{\sigma(\varphi_m+v_m)} \cdots \underline{\alpha}_{\sigma(t)}}} e^{-\frac{1}{2}\left (\sum_{i=1}^{n_s} \omega_{1_i} | \mathcal{O}_i (\varphi_m+v_m,t)|\right )} \|x(\varphi_m+v_m)\| \notag \\ &+ \sqrt{\frac{(\varepsilon_1+\varepsilon_2+\psi_1) \overline{\alpha}}{\zeta_1 \underline{\alpha}}+G_1}{\|w_t\|}_\infty,
\end{align}
$\forall m\in\mathbb{N}$ and $\forall t\in Y_m$. All relevant constants are the same as in the proof of Theorem \ref{th1}.\\
Consider the intervals $W_m,$ $m \in \mathbb{N}$, across which the inequality \eqref{eq:update_rule_1} may not hold. In this situation, certain intermediary procedures are required to derive a constraint on the development of $V(x(t))$. First, we have to prove that each $m\in\mathbb{N}$
\begin{align}
    \|\xi_\sigma(t)\| \leq & \left( \frac{1+\sqrt{2}}{2} \sqrt{\psi_1} +1 \right ) \|y_\sigma (\varphi_m)\| +\frac{1+\sqrt{2}}{2} \sqrt{\psi_2} \|x_e(\varphi_m)\| + \|y_\sigma (t)\| \notag\\
    & + \frac{1+\sqrt{2}}{2} \sqrt{\psi_1} \|w_{\varphi_m}\|_\infty \label{eq:xi_sigma_2} \\
    \|\xi_e(t)\| \leq & \frac{1+\sqrt{2}}{2} \sqrt{\psi_1} \|y_\sigma(\varphi_m)\| + \left (\frac{1+\sqrt{2}}{2}\sqrt{\psi_2}+1 \right) \|x_e(\varphi_m)\|+ \|x_e(t)\| \notag \\
    &+\frac{1+\sqrt{2}}{2} \sqrt{\psi_1} \|w_{\varphi_m}\|_\infty \label{eq:xi_e_2}
\end{align}
$\forall t\in W_m$. From \eqref{eq:update_rule_1}, we can write
\begin{align}
    \|\xi_\sigma (\varphi_m)\|^2 +\|\xi_e (\varphi_m)\|^2 \leq \psi_1 \|y_\sigma (\varphi_m)\|^2 +\psi_2 \|x_e(\varphi_m)\|^2 +\psi_1 \|w_{\varphi_m}\|_\infty^2. \label{eq:error_fsdos}
\end{align}
From \eqref{eq:error_fsdos}, we can write
\begin{align}
    2\|\xi_\sigma (\varphi_m)\|^2 +2\|\xi_e (\varphi_m)\|^2 \leq 2\psi_1 \|y_\sigma (\varphi_m)\|^2 +2\psi_2 \|x_e(\varphi_m)\|^2 +2\psi_1 \|w_{\varphi_m}\|_\infty^2. \label{eq:error_fsdos_3}
\end{align}
From \eqref{eq:error_fsdos_3}, we can achieve
\begin{align}
    & \|\xi_\sigma (\varphi_m)\|^2 +\|\xi_e (\varphi_m) \|^2 +2 \|\xi_\sigma (\varphi_m) \| \|\xi_e (\varphi_m)\| \leq 2 \{ \psi_1 \|y_\sigma (\varphi_m) \|^2 + \psi_2 \|x_e(\varphi_m)\|^2 \notag \\
    & + \psi_1 \|w_{\varphi_m}\|_\infty^2+ 2 \sqrt{\psi_1} \|y_\sigma (\varphi_m) \| \sqrt{\psi_2} \|x_e(\varphi_m)\|+ 2\sqrt{\psi_2} \|x_e(\varphi_m)\| \sqrt{\psi_1} \|w_{\varphi_m}\|_\infty \notag \\
    & + 2 \sqrt{\psi_1} \|w_{\varphi_m}\|_\infty \sqrt{\psi_1}\|y_\sigma(\varphi_m)\| \}. 
\end{align}
Hence,
\begin{align}
    \|\xi_\sigma(\varphi_m)\|+ \|\xi_e (\varphi_m)\| \leq \sqrt{2 \psi_1} ( \|y_\sigma (\varphi_m)\| + \|w_{\varphi_m}\|_\infty ) +\sqrt{2 \psi_2} \|x_e (\varphi_m)\|. \label{eq:error_fsdod_4}
\end{align}
From \eqref{eq:error_fsdos}, we can reach
\begin{align}
    &\|\xi_\sigma (\varphi_m)\|^2 +\|\xi_e (\varphi_m)\|^2 -2 \|\xi_\sigma (\varphi_m)\| \|\xi_e(\varphi_m) \| \leq \psi_1 \|y_\sigma (\varphi_m) \|^2 + \psi_2 \|x_e (\varphi_m)\|^2 \notag \\
    &+\psi_1 \|w_{\varphi_m} \|^2_\infty +2\sqrt{\psi_1} \|y_\sigma (\varphi_m)\| \sqrt{\psi_2} \|x_e (\varphi_m)\| +2\sqrt{\psi_2} \|x_e(\varphi_m)\| \sqrt{\psi_1} \|w_{\varphi_m}\|_\infty \notag \\
    & +2 \sqrt{\psi_1} \|w_{\varphi_m}\|_\infty \sqrt{\psi_1} \|y_\sigma(\varphi_m)\|.
\end{align}
Hence,
\begin{align}
    \|\xi_\sigma (\varphi_m) \|- \|\xi_e (\varphi_m)\| \leq \sqrt{\psi_1} \left ( \|y_\sigma (\varphi_m)\|+ \|w_{\varphi_m}\|_\infty \right )+ \sqrt{\psi_2} \|x_e (\varphi_m)\| \label{eq:error_fsdos_1}
\end{align}
and
\begin{align}
    \|\xi_e (\varphi_m)\|- \|\xi_\sigma (\varphi_m) \| \leq \sqrt{\psi_1} \left ( \|y_\sigma (\varphi_m)\|+ \|w_{\varphi_m}\|_\infty \right )+ \sqrt{\psi_2} \|x_e (\varphi_m)\|. \label{eq:error_fsdos_2}
\end{align}
By adding \eqref{eq:error_fsdod_4} and \eqref{eq:error_fsdos_1}, we get
\begin{align}
    \|\xi_\sigma (\varphi_m) \| \leq \frac{1+\sqrt{2}}{2} \sqrt{\psi_1} (\|y_\sigma (\varphi_m)\|+ \|w_{\varphi_m}\|_\infty) + \frac{1+\sqrt{2}}{2} \sqrt{\psi_2} \|x_e (\varphi_m)\| \label{eq:error_fsdos_5}
\end{align}
and by adding \eqref{eq:error_fsdod_4} and \eqref{eq:error_fsdos_2}, we get
\begin{align}
    \|\xi_e (\varphi_m)\| \leq \frac{1+\sqrt{2}}{2} \sqrt{\psi_1} (\|y_\sigma (\varphi_m)\|+ \|w_{\varphi_m}\|_\infty) + \frac{1+\sqrt{2}}{2} \sqrt{\psi_2} \|x_e (\varphi_m)\|. \label{eq:error_fsdos_6}
\end{align}
Hence from \eqref{eq:error_fsdos_5},
\begin{align}
    \|y_\sigma (t_{k(\varphi_m)})- y_\sigma (\varphi_m)\| \leq &\frac{1+\sqrt{2}}{2} \sqrt{\psi_1} (\|y_\sigma (\varphi_m)\|+ \|w_{\varphi_m}\|_\infty) \notag \\
    & + \frac{1+\sqrt{2}}{2} \sqrt{\psi_2} \|x_e (\varphi_m)\|
\end{align}
and \eqref{eq:xi_sigma_2} follows by applying the triangular inequality. From \eqref{eq:error_fsdos_6}, we are able to reach out
\begin{align}
    \|x_e(t_{k(\varphi_m)}) -x_e(\varphi_m) \| \leq &\frac{1+\sqrt{2}}{2} \sqrt{\psi_1} (\|y_\sigma (\varphi_m)\|+ \|w_{\varphi_m}\|_\infty) \notag \\ 
    & + \frac{1+\sqrt{2}}{2} \sqrt{\psi_2} \|x_e (\varphi_m)\|
\end{align}
and we get \eqref{eq:xi_e_2} by applying triangular inequality. From \eqref{eq:V_dot_1}, we get
\begin{align}
    &\frac{d}{dt}V_\sigma(x(t), e_\sigma(t)) \leq x^\top(t) \Gamma_{5_\sigma} x(t) -2x^\top_p(t)P_{p_\sigma}BK\xi_e(t) -2\tilde{x}^\top(t) P_{e_\sigma} L_\sigma \xi_\sigma (t) \notag \\
    & +(\varepsilon_{1_\sigma}+ \varepsilon_{2_\sigma}) w^\top(t)w(t), \label{eq:V_dot_fsdos_21}
\end{align}
where $ \Gamma_{5_\sigma}, \varepsilon_{1_\sigma},$ and $\varepsilon_{2_\sigma}$ are mentioned in \eqref{eq:Gamma_5} and \eqref{eq:V_dot_6}. If $\Gamma_{2_\sigma}$ is positive definite then $\Gamma_{5_\sigma}$ is negative definite. Hence,
\begin{align}
    & \frac{d}{dt} V_\sigma (x(t), e_\sigma(t)) \leq -\gamma_{1_\sigma}\|x(t)\|^2 +2 \gamma_{2_\sigma}\|x_p(t)\| \|\xi_e(t)\| +2 \gamma_{3_\sigma} \|\tilde{x}(t)\| \|\xi_\sigma(t)\| \notag \\
    & +(\varepsilon_{1_\sigma}+ \varepsilon_{2_\sigma}) \|w(t)\|^2 \label{eq:V_dot_fsdos_1}
\end{align}
where $\gamma_{1_\sigma}, \gamma_{2_\sigma},$ and $\gamma_{3_\sigma}$ are mentioned in \eqref{eq:Gamma_3}. From \eqref{eq:xi_sigma_2}, and \eqref{eq:xi_e_2} we have
\begin{align}
    &\|\xi_\sigma (t)\| \leq \left \{ \left ( \frac{1+\sqrt{2}}{2} \sqrt{\psi_1}+1 \right ) \|C_\sigma\| +\frac{1+\sqrt{2}}{2}\sqrt{\psi_2} \right \} \|x_p(\varphi_m)\| + \frac{1+\sqrt{2}}{2} \notag \\
    &  \sqrt{\psi_2}\|\tilde{x}(\varphi_m)\| +\|C_\sigma\| \|x_p(t)\| + \frac{1+\sqrt{2}}{2} \sqrt{\psi_1} \|w_{\varphi_m}\|_\infty \label{eq:xi_sigma_3}
\end{align}
and
\begin{align}
    & \|\xi_e (t)\| \leq \left ( \frac{1+\sqrt{2}}{2} \sqrt{\psi_1} \|C_\sigma\| + \frac{1+\sqrt{2}}{2} \sqrt{\psi_2} +1 \right ) \|x_p(\varphi_m)\| + \Biggl ( \frac{1+\sqrt{2}}{2} \sqrt{\psi_2} \notag \\
    & +1 \Biggr ) \|\tilde{x} (\varphi_m)\| +\|x_p(t)\| +\|\tilde{x}(t)\| + \frac{1+\sqrt{2}}{2} \sqrt{\psi_1} \|w_{\varphi_m}\|_\infty, \label{eq:xi_e_3}
\end{align}
respectively, by using \eqref{eq:x_tilde} and $y_\sigma (t)=C_\sigma x_p(t)$. Using \eqref{eq:xi_sigma_3} and \eqref{eq:xi_e_3}, in \eqref{eq:V_dot_fsdos_1}, we get
\begin{align}
    &\frac{d}{dt} V_\sigma (x_p(t), \tilde{x}(t), x_p(\varphi_m), \tilde{x}(\varphi_m), f(t)) \leq (2 \gamma_{2_\sigma} - \gamma_{1_\sigma}) \|x_p(t)\|^2 -\gamma_{1_\sigma} \|\tilde{x}(t) \|^2 \notag \\
    & +2 \gamma_{3_\sigma} \|C_\sigma\| \|x_p(t)\| \|\tilde{x}(t)\| + \gamma_{2_\sigma} \{(1+\sqrt{2}) \sqrt{\psi_1} \|C_\sigma\| +(1+\sqrt{2}) \sqrt{\psi_2} +2\} \notag \\
    & \|x_p(t)\| \|x_p(\varphi_m)\| + \gamma_{2_\sigma} \{(1+\sqrt{2}) \sqrt{\psi_2}+2\} \|x_p(t)\| \|\tilde{x}(\varphi_m)\| + \gamma_{3_\sigma} [\{ (1+\sqrt{2}) \notag \\
    & \sqrt{\psi_1}+2 \} \|C_\sigma\| +(1+\sqrt{2}) \sqrt{\psi_2}] \|\tilde{x}(t)\| \|x_p(\varphi_m)\| +\gamma_{3_\sigma} (1+\sqrt{2}) \sqrt{\psi_2} \|\tilde{x}(t)\| \notag \\
    & \|\tilde{x}(\varphi_m)\| +\gamma_{2_\sigma} (1+\sqrt{2}) \sqrt{\psi_1} \|x_p(t)\| \|w_{\varphi_m}\|_\infty +\gamma_{3_\sigma} (1+\sqrt{2}) \sqrt{\psi_1} \|\tilde{x}(t)\| \|w_{\varphi_m}\|_\infty \notag \\
    & +(\varepsilon_{1_\sigma}+ \varepsilon_{2_\sigma}) \|w(t)\|^2, \label{eq:V_dot_fsdos_2}
\end{align}
where $f(t)$ is mentioned in \eqref{eq:f}. By using Young's inequality, we have
\begin{align}
    \gamma_{2_\sigma} (1+\sqrt{2}) \sqrt{\psi_1} \|x_p(t)\| \|w_{\varphi_m}\|_\infty \leq \gamma_{2_\sigma} \frac{1+ \sqrt{2}}{2 \varepsilon_{3_\sigma}} \sqrt{\psi_1} \|x_p(t)\|^2 +\frac{\varepsilon_{3_\sigma}}{2}  \|w_{\varphi_m}\|^2_\infty \label{eq:yi_1}
\end{align}
and
\begin{align}
    \gamma_{3_\sigma} (1+\sqrt{2}) \sqrt{\psi_1} \|\tilde{x}(t)\| \|w_{\varphi_m}\|_\infty \leq \gamma_{3_\sigma} \frac{1+\sqrt{2}}{2 \varepsilon_{4_\sigma}} \sqrt{\psi} \|\tilde{x}(t)\|^2
    + \frac{\varepsilon_{4_\sigma}}{2} \|w_{\varphi_m}\|^2(t). \label{eq:yi_2}
\end{align}
If we choose $\varepsilon_{3_\sigma}$ and $\varepsilon_{4_\sigma}$ very high then the values of $\gamma_{2_\sigma} \frac{1+ \sqrt{2}}{2 \varepsilon_{3_\sigma}} \sqrt{\psi_1} \|x_p(t)\|^2$ and $\gamma_{3_\sigma} \frac{1+\sqrt{2}}{2 \varepsilon_{4_\sigma}} \sqrt{\psi} \|\tilde{x}(t)\|^2$ are negligible. Using \eqref{eq:yi_1} and \eqref{eq:yi_2} in \eqref{eq:V_dot_fsdos_2}, we get
\begin{align}
    \frac{d}{dt} V_\sigma (X_f(t)) \leq &X^\top_f(t) \Gamma_{3_\sigma} X_f(t) +X^\top_f(t) \Gamma_{4_\sigma} X_f(\varphi_m) + \nu_{2_\sigma} f^2(t), 
\end{align}
where $\Gamma_{3_\sigma}$ and $\Gamma_{4_\sigma}$ are mentioned in \eqref{eq:Gamma_3} and \eqref{eq:Gamma_4},
\begin{align}
    X_f(t) \coloneqq \begin{bmatrix}
        \|x_p(t)\| & \|\tilde{x}(t)\|
    \end{bmatrix}^\top ,
\end{align}
and
\begin{align}
     \nu_{2_\sigma}\coloneqq \varepsilon_{1_\sigma}+ \varepsilon_{2_\sigma} + \frac{\varepsilon_{3_\sigma}}{2} +\frac{\varepsilon_{4_\sigma}}{2}.
\end{align}
Hence,
\begin{align}
    \frac{d}{dt} V_\sigma (x(t)) \leq \omega_{2_\sigma} \max \{ V_\sigma(x(t)), V_\sigma (x(\varphi_m)\} + \nu_{2_\sigma} f^2(t),
\end{align}
where $\omega_{2_\sigma}$ is mentioned in \eqref{eq:omega_2_sigma}. Thus, standard comparison results of differential inequalities yields
\begin{align}
    V_\sigma(x(t)) \leq e^{\omega_{2_\sigma} (t- \varphi_m)} V_\sigma (x(\varphi_m)) + \nu_{3_\sigma} e^{\omega_{2_\sigma} (t- \varphi_m)} \|w_t\|_\infty^2, \label{eq:V_fsdos}
\end{align}
where $\nu_{3_\sigma} \coloneqq \nu_{2_\sigma}/\omega_{2_\sigma}, \forall t\in W_m$. Consider the interval $Y_m$ where \eqref{eq:update_rule_1} holds true by construction and contains $q_m$ jump points. $g(t)$ is total number of jump points, so $q_m \leq g(t)$. In $W_m$, the starting point $\varphi_m$ may not be a jump point, but the ending point $\varphi_m+v_m$ may be a jump point where the mode of the observer is going to change. From \eqref{eq:V_3}, we can write
\begin{align}
    V_\sigma (x(\varphi_{m(t)})) \leq e^{-\omega_{1_\sigma} (\varphi_{m(t)} -\iota_{g(t)})} V_\sigma (x(\iota_{g(t)})) + \nu_{1_\sigma} \|w_t\|^2_\infty \label{eq:V_fsdos_2}
\end{align}
$\forall t \in [\varphi_{m(t)}, \iota_{g(t)}]$ where
\begin{align}
    m(t) \coloneqq \begin{cases}
    0, & \text{if } t< \varphi_1 \\
    \sup\{m \in \mathbb{N}: \varphi_m \leq t\}, & \text{otherwise.}
    \end{cases} \label{eq:m_t}
\end{align}
The sum of all the jump points in $Y_m$ equals to $g(t)$. Hence, from Figure \ref{fig:time_axis}, we can say 
\begin{align}
    \sum_{s=1}^{m(t)} q_{m(t)-s} =g(t).
\end{align}
Using the upper bound of $V_\sigma (x(\varphi_{m(t)}))$ in \eqref{eq:V_fsdos}, we get
\begin{align}
    V_\sigma (x(t)) \leq e^{-\omega_{1_\sigma}|\Bar{\Upsilon}(\iota_{g(t)},t)|} e^{\omega_{2_\sigma}|\Bar{\Omega} (\iota_{g(t)},t)|} V_\sigma (x(\iota_{g(t)}))+ (\nu_{1_\sigma} +\nu_{3_\sigma}) e^{\omega_{2_\sigma} (t-\varphi_{m(t)})} \|w_t\|^2_\infty \label{eq:V_fsdos_3}
\end{align}
$\forall t \in [\iota_{g(t)}, \varphi_{m(t)}+v_{m(t)}).$ Using \eqref{eq:V_2}, we get
\begin{align}
    \|x(t)\|^2 \leq & \frac{\overline{\alpha}_{\sigma(\iota^+_{g(t)})}}{\underline{\alpha}_{\sigma(\iota^+_{g(t)})}} e^{-\omega_{1_\sigma}|\Bar{\Upsilon}(\iota_{g(t)},t)|} e^{\omega_{2_\sigma}|\Bar{\Omega} (\iota_{g(t)},t)|} \|x(\iota_{g(t)})\|^2 + \frac{\nu_{1_\sigma} +\nu_{3_\sigma}}{\underline{\alpha}_{\sigma(\iota^+_{g(t)})}} \notag \\
    & e^{\omega_{2_\sigma} (t-\varphi_{m(t)})} \|w_t\|^2_\infty \label{eq:trajectory_fsdos}
\end{align}
$\forall t\in [\iota_{g(t)}, \varphi_{m(t)}+v_{m(t)})$. We can establish the following results by combining \eqref{eq:trajectory_1} and \eqref{eq:trajectory_fsdos} and some iterative process.
\begin{figure}
    \centering
    \includegraphics[width=\textwidth, trim={0.5cm 5.1cm 0 7.55cm}, clip]{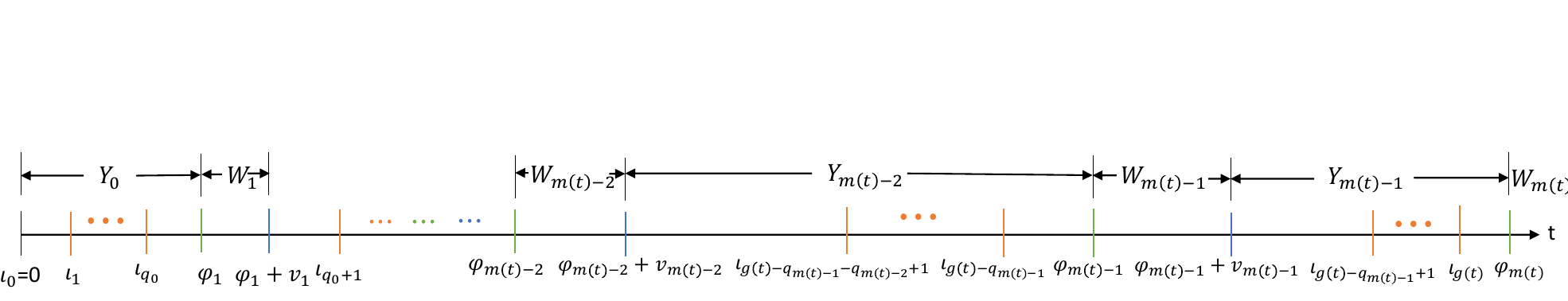}
    \caption{Time axis from $0$ to $t$ where red lines represent jump points, green lines represent the starting points of the interval $W_m,m \in \mathbb{N}$ where \eqref{eq:update_rule_1} may not hold, and blue lines represent the ending point of the interval $W_m, m \in \mathbb{N}$. Green lines are not jump points, but some blue lines may be jump points. This blue line and the next red line will be considered a single point in that scenario. For example, if $\varphi_1 + v_1$ is a jump point then $\varphi_1 + v_1 = \iota_{q_0+1}$.}
    \label{fig:time_axis}
\end{figure}
\begin{figure}
    \centering
    \includegraphics[width=\textwidth, trim={0.1cm 4.3cm 0 7.55cm}, clip]{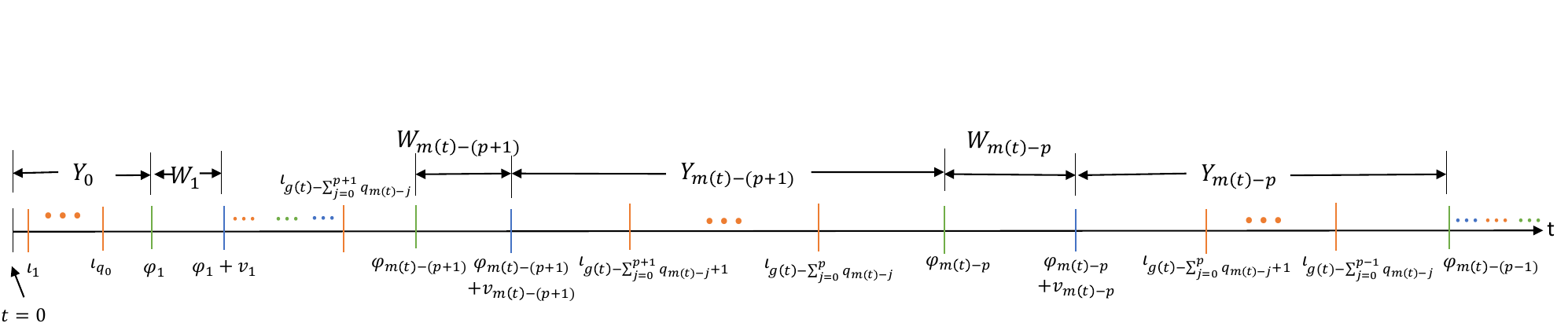}
    \caption{Time axis from $0$ to $t$ where $W_{m(t)-(p+1)}, Y_{m(t)-(p+1)}, W_{m(t)-p},$ and $Y_{m(t)-p}$ sections are highlighted and red lines represent jump points, green lines represent the starting points of the interval $W_m,m \in \mathbb{N}$ where \eqref{eq:update_rule_1} may not hold, and blue lines represent the ending point of the interval $W_m, m \in \mathbb{N}$.}
    \label{fig:time_axis_pth}
\end{figure}
\begin{lemma}\label{lm:trajectory_fsdos}
    For all $t \in [0,\varphi_{m(t)}+v_{m(t)}),$ we get
    \begin{align}
        &\|x(t)\|^2 \leq \frac{\overline{\alpha}_{\sigma(\iota^+_0)} \cdots \overline{\alpha}_{\sigma(\iota^+_{g(t)})}}{\underline{\alpha}_{\sigma(\iota^+_0)} \cdots \underline{\alpha}_{\sigma(\iota^+_{g(t)})}} e^{-\sum_{i=1}^{n_s} \omega_{1_i}|(\mathcal{O}_i\cap\overline{\Upsilon})(\iota_0,t)|} e^{\sum_{i=1}^{n_s} \omega_{2_i}|(\mathcal{O}_i\cap \overline{\Omega})(\iota_0,t)|} \|x(\iota_0)\|^2 \notag \\
        & +\frac{1}{\underline{\alpha}_{\sigma(\iota^+_{g(t)})}} \Biggl [ \mathlarger{\sum}_{\substack{{l \in \mathbb{N};} \\ {0 \leq \iota_{g(t)-l} < \iota_{g(t)}}}} \Biggl \{ \nu_{1_{\sigma(\iota^+_{g(t)-l})}} \mathlarger{\prod}_{\substack{{j \in \mathbb{N}_0;} \\ {\iota_{g(t)-l} < \iota_{g(t)-j} \leq \iota_{g(t)}}}} \left ( \frac{\overline{\alpha}_{\sigma(\iota^+_{g(t)-j})}}{\underline{\alpha}_{\sigma(\iota^+_{g(t)-(j+1)})}} \right ) \notag \\
        & e^{-(\sum_{i=1}^{n_s} \omega_{1_i} |(\mathcal{O}_i \cap \overline{\Upsilon})(\iota_{g(t)-(l-1)},t)|)} e^{ (\sum_{i=1}^{n_s} \omega_{2_i} |(\mathcal{O}_i \cap \overline{\Omega}) (\iota_{g(t)-l},t)|)} \Biggr \} + \mathlarger{\sum}_{\substack{{k \in \mathbb{N}_0;}\\{q\coloneqq \sum_{s=0}^k q_{m(t)-s} \leq g(t)}}} \notag \\
        & \Biggl \{ \left(\nu_{1_{\sigma(\iota^+_{g(t)-q})}} + \nu_{3_{\sigma(\iota^+_{g(t)-q})}} \right ) \mathlarger{\prod}_{\substack{{j \in \mathbb{N}_0;}\\{\iota_{g(t)-q} < \iota_{g(t)-j} \leq \iota_{g(t)}}}} \left( \frac{\overline{\alpha}_{\sigma(\iota^+_{g(t)-j})}}{\underline{\alpha}_{\sigma(\iota^+_{g(t)-(j+1)})}} \right )  \notag \\
        & e^{-(\sum_{i=1}^{n_s} \omega_{1_i}|(\mathcal{O}_i \cap \overline{\Upsilon})(\iota_{g(t)-(q-1)},t)|)} e^{(\sum_{i=1}^{n_s} \omega_{2_i}|(\mathcal{O}_i \cap \overline{\Omega})(\iota_{g(t)-q},t)|)} \Biggr \} \Biggr ] \|w_t\|^2_\infty, \label{eq:trajectory_fsdos_14}
    \end{align}
    by considering $q_{m(t)}=0$ (because $Y_{m(t)}$ does not exist in Figure \ref{fig:time_axis}) and $|(\mathcal{O}_i \cap \overline{\Upsilon})(\iota_{g(t)+1},t)|=0$ (because $t < \iota_{g(t)+1}$ in Figure \ref{fig:time_axis}).
\end{lemma}
$Proof$ $of$ $Lemma$ $\ref{lm:trajectory_fsdos}$ $:$
We use an induction argument to prove \eqref{eq:trajectory_fsdos_14}. First, the inequality holds true over $Y_0=[0, \varphi_1)$ is shown. If $\varphi_1=0,$ the claim trivially holds. Suppose $\varphi_1>0$. Over $Y_0$ the $\|x(t)\|^2$ obeys \eqref{eq:trajectory_2}. Thus \eqref{eq:trajectory_fsdos_14} follows by noting that $|(\mathcal{O}_i \cap \overline{\Upsilon})(0,t)|=|\mathcal{O}_i(0,t)|$ and $|(\mathcal{O}_i \cap \overline{\Omega})(0,t)|=0,$ for all $t \in Y_0$. \\
Consider interval $Y_m$ where \eqref{eq:update_rule_1} holds true, and it contains $q_m$ jump points (time instants when observer switching takes place). Lets consider $Y_{m(t)-1}$ interval which contains $q_{m(t)-1}$ jump points and last jump point is $\iota_{g(t)}$ shown in Figure \ref{fig:time_axis}. So following \eqref{eq:trajectory_2}, we can write
\begin{align}
    &\|x(\iota_{g(t)})\|^2 \leq \frac{\overline{\alpha}_{\sigma(\iota^+_{g(t)-1})} \cdots \overline{\alpha}_{\sigma(\iota^+_{g(t)-q_{m(t)-1}+1})}}{\underline{\alpha}_{\sigma(\iota^+_{g(t)-1})} \cdots \underline{\alpha}_{\sigma(\iota^+_{g(t)-q_{m(t)-1}+1})}} e^{-\left (\sum_{i=1}^{n_s} \omega_{1_i} | \mathcal{O}_i (\iota_{g(t)-q_{m(t)-1}+1},\iota_{g(t)})|\right )} \notag \\
    & \|x(\iota_{g(t)-q_{m(t)-1}+1})\|^2 + \Biggl [ \frac{\nu_{1_{\sigma(\iota^+_{g(t)-1})}}}{\underline{\alpha}_{\sigma(\iota^+_{g(t)-1})}} + \mathlarger{\mathlarger{\sum}}_{\substack{l \in \mathbb{N}\setminus 1;\\ \iota_{g(t)-q_{m(t)-1}+1} \leq \iota_{g(t)-l} < \iota_{g(t)}}}\Biggl \{ \frac{\nu_{1_{\sigma(\iota^+_{g(t)-l})}}}  { \underline{\alpha}_{\sigma(\iota^+_{g(t)-1})}} \notag\\ 
    & \prod_{\substack{j \in \mathbb{N}; \\ \iota_{g(t)-l}< \iota_{g(t)-j} < \iota_{g(t)}}} \Bigl ( \frac{ \overline{\alpha}_{\sigma(\iota^+_{g(t)-j})}} { \underline{\alpha}_{\sigma(\iota^+_{g(t)-(j+1)})}} \Bigr ) e^{-\left ( \sum_{i=1}^{n_s} \omega_{1_i} |\mathcal{O}_i(\iota_{g(t)-(l-1)}, \iota_{g(t)})| \right) }  \Biggr \} \Biggr ]  {\|w_t\|}^2_\infty. \label{eq:trajectory_fsdos_1}
\end{align}
Using the inequality of \eqref{eq:trajectory_fsdos_1}, in \eqref{eq:trajectory_fsdos}, we get
\begin{align}
    & \|x(t)\|^2 \leq \frac{\overline{\alpha}_{\sigma(\iota^+_{g(t)})} \cdots \overline{\alpha}_{\sigma(\iota^+_{g(t)-q_{m(t)-1}+1})}}{\underline{\alpha}_{\sigma(\iota^+_{g(t)})} \cdots \underline{\alpha}_{\sigma(\iota^+_{g(t)-q_{m(t)-1}+1})}} e^{-\left (\sum_{i=1}^{n_s} \omega_{1_i} | (\mathcal{O}_i \cap \overline{\Upsilon}) (\iota_{g(t)-q_{m(t)-1}+1},t)|\right )} \notag \\
    & e^{\left (\sum_{i=1}^{n_s} \omega_{2_i} |(\mathcal{O}_i \cap \overline{\Omega}) (\iota_{g(t)-q_{m(t)-1}+1},t)| \right)} \|x(\iota_{g(t)-q_{m(t)-1}+1})\|^2 + \frac{1}{\underline{\alpha}_{\sigma(\iota^+_{g(t)})}} \Biggl [ (\nu_{1_{\sigma(\iota^+_{g(t)})}} \notag \\
    & +\nu_{3_{\sigma(\iota^+_{g(t)})}} ) e^{\omega_{2_{\sigma(\iota^+_{g(t)})}} (t-\varphi_{m(t)})} + \mathlarger{\mathlarger{\sum}}_{\substack{l \in \mathbb{N};\\ \iota_{g(t)-q_{m(t)-1}+1} \leq \iota_{g(t)-l} < \iota_{g(t)}}}\Biggl \{ \nu_{1_{\sigma(\iota^+_{g(t)-l})}} \notag \\
    & \prod_{\substack{j \in \mathbb{N}_0; \\ \iota_{g(t)-l}< \iota_{g(t)-j} \leq \iota_{g(t)}}} \left ( \frac{ \overline{\alpha}_{\sigma(\iota^+_{g(t)-j})}} { \underline{\alpha}_{\sigma(\iota^+_{g(t)-(j+1)})}} \right ) e^{-\left ( \sum_{i=1}^{n_s} \omega_{1_i} |(\mathcal{O}_i \cap \overline{\Upsilon})(\iota_{g(t)-(l-1)}, t)| \right) } \notag \\
    & e^{\left ( \sum_{i=1}^{n_s} \omega_{2_i} |(\mathcal{O}_i \cap \overline{\Omega})(\iota_{g(t)-(l-1)}, t)| \right) }  \Biggr \} \Biggr ]  {\|w_t\|}^2_\infty, \label{eq:trajectory_fsdos_2}
\end{align}
$\forall t \in \mathbb{R}_{\geq \iota_{g(t)-q_{m(t)-1}+1}}$ and consider that $|\overline{\Omega}(\iota_{g(t)-q_{m(t)-1}+1}, \iota_{g(t)})|=0$. From \eqref{eq:V_fsdos_3}, we can write
\begin{align}
    & V_{\sigma(\iota^+_{g(t)-q_{m(t)-1}})} (x(\varphi_{m(t)-1}+v_{m(t)-1})) \leq e^{- \left (\omega_{1_{\sigma(\iota^+_{g(t)-q_{m(t)-1}})}}|\Bar{\Upsilon}(\iota_{g(t)-q_{m(t)-1}},\varphi_{m(t)-1}+v_{m(t)-1})| \right )} \notag \\
    & e^{\left (\omega_{2_{\sigma(\iota^+_{g(t)-q_{m(t)-1}})}}|\Bar{\Omega} (\iota_{g(t)-q_{m(t)-1}},\varphi_{m(t)-1}+v_{m(t)-1})| \right )} V_{\sigma(\iota^+_{g(t)-q_{m(t)-1}})} (x(\iota_{g(t)-q_{m(t)-1}})) \notag \\
    & + \left (\nu_{1_{\sigma(\iota^+_{g(t)-q_{m(t)-1}})}} +\nu_{3_{\sigma(\iota^+_{g(t)-q_{m(t)-1}})}} \right ) e^{\omega_{2_{\sigma(\iota^+_{g(t)-q_{m(t)-1}})}} v_{m(t)-1}} \|w_t\|^2_\infty. \label{eq:V_fsdos_4}
\end{align}
From \eqref{eq:V_fsdos_2}, we can write,
\begin{align}
    & V_{\sigma(\iota^+_{g(t)-q_{m(t)-1}})} (x(\iota_{g(t)-q_{m(t)-1}+1})) \leq e^{-\omega_{1_{\sigma(\iota^+_{g(t)-q_{m(t)-1}})}} (\iota_{g(t)-q_{m(t)-1}+1}-(\varphi_{m(t)-1}+v_{m(t)-1}))} \notag \\
    & V_{\sigma(\iota^+_{g(t)-q_{m(t)-1}})} (x(\varphi_{m(t)-1}+v_{m(t)-1})) + \nu_{1_{\sigma(\iota^+_{g(t)-q_{m(t)-1}})}} \|w_t\|^2_\infty. \label{eq:V_fsdos_5}
\end{align}
Using the inequality of \eqref{eq:V_fsdos_4} in \eqref{eq:V_fsdos_5}, we get
\begin{align}
    & V_{\sigma(\iota^+_{g(t)-q_{m(t)-1}})} (x(\iota_{g(t)-q_{m(t)-1}+1})) \leq e^{- \left (\omega_{1_{\sigma(\iota^+_{g(t)-q_{m(t)-1}})}}|\Bar{\Upsilon}(\iota_{g(t)-q_{m(t)-1}},\iota_{g(t)-q_{m(t)-1}+1})| \right )} \notag \\
    & e^{\left (\omega_{2_{\sigma(\iota^+_{g(t)-q_{m(t)-1}})}}|\Bar{\Omega} (\iota_{g(t)-q_{m(t)-1}},\iota_{g(t)-q_{m(t)-1}+1})| \right )} V_{\sigma(\iota^+_{g(t)-q_{m(t)-1}})} (x(\iota_{g(t)-q_{m(t)-1}})) \notag \\
    & + \Biggl [ \nu_{1_{\sigma(\iota^+_{g(t)-q_{m(t)-1}})}} e^{- \left (\omega_{1_{\sigma(\iota^+_{g(t)-q_{m(t)-1}})}}|\Bar{\Upsilon}(\iota_{g(t)-q_{m(t)-1}},\varphi_{m(t)-1}+v_{m(t)-1})| \right )} \notag \\
    & e^{\left (\omega_{2_{\sigma(\iota^+_{g(t)-q_{m(t)-1}})}}|\Bar{\Omega} (\iota_{g(t)-q_{m(t)-1}},\varphi_{m(t)-1}+v_{m(t)-1})| \right )} + \left (\nu_{1_{\sigma(\iota^+_{g(t)-q_{m(t)-1}})}} +\nu_{3_{\sigma(\iota^+_{g(t)-q_{m(t)-1}})}} \right ) \notag \\
    & e^{\omega_{2_{\sigma(\iota^+_{g(t)-q_{m(t)-1}})}} v_{m(t)-1}} \Biggr ] \|w_t\|^2_\infty. \label{eq:V_fsdos_6}
\end{align}
Hence,
\begin{align}
    & \|x(\iota_{g(t)-q_{m(t)-1}+1})\|^2 \leq \frac{\overline{\alpha}_{\sigma(\iota^+_{g(t)-q_{m(t)-1}})}}{\underline{\alpha}_{\sigma(\iota^+_{g(t)-q_{m(t)-1}})}} e^{- \left (\omega_{1_{\sigma(\iota^+_{g(t)-q_{m(t)-1}})}}|\Bar{\Upsilon}(\iota_{g(t)-q_{m(t)-1}},\iota_{g(t)-q_{m(t)-1}+1})| \right )} \notag \\
    & e^{\left (\omega_{2_{\sigma(\iota^+_{g(t)-q_{m(t)-1}})}}|\Bar{\Omega} (\iota_{g(t)-q_{m(t)-1}},\iota_{g(t)-q_{m(t)-1}+1})| \right )} \|x(\iota_{g(t)-q_{m(t)-1}})\|^2+\frac{1}{\underline{\alpha}_{\sigma(\iota^+_{g(t)-q_{m(t)-1}})}} \notag \\
    & \Biggl [ \nu_{1_{\sigma(\iota^+_{g(t)-q_{m(t)-1}})}} e^{- \left (\omega_{1_{\sigma(\iota^+_{g(t)-q_{m(t)-1}})}}|\Bar{\Upsilon}(\iota_{g(t)-q_{m(t)-1}},\varphi_{m(t)-1}+v_{m(t)-1})| \right )} \notag \\
    & e^{\left (\omega_{2_{\sigma(\iota^+_{g(t)-q_{m(t)-1}})}}|\Bar{\Omega} (\iota_{g(t)-q_{m(t)-1}},\varphi_{m(t)-1}+v_{m(t)-1})| \right )} + \left (\nu_{1_{\sigma(\iota^+_{g(t)-q_{m(t)-1}})}} +\nu_{3_{\sigma(\iota^+_{g(t)-q_{m(t)-1}})}} \right ) \notag \\
    & e^{\omega_{2_{\sigma(\iota^+_{g(t)-q_{m(t)-1}})}} v_{m(t)-1}} \Biggr ] \|w_t\|^2_\infty. \label{eq:trajectory_fsdos_3}
\end{align}
Using the inequality of \eqref{eq:trajectory_fsdos_3} in \eqref{eq:trajectory_fsdos_2}, we get
\begin{align}
    & \|x(t)\|^2 \leq \frac{\overline{\alpha}_{\sigma(\iota^+_{g(t)})} \cdots \overline{\alpha}_{\sigma(\iota^+_{g(t)-q_{m(t)-1}})}}{\underline{\alpha}_{\sigma(\iota^+_{g(t)})} \cdots \underline{\alpha}_{\sigma(\iota^+_{g(t)-q_{m(t)-1}})}} e^{-\left (\sum_{i=1}^{n_s} \omega_{1_i} | (\mathcal{O}_i \cap \overline{\Upsilon}) (\iota_{g(t)-q_{m(t)-1}},t)|\right )} \notag \\
    & e^{\left (\sum_{i=1}^{n_s} \omega_{2_i} |(\mathcal{O}_i \cap \overline{\Omega}) (\iota_{g(t)-q_{m(t)-1}},t)| \right)} \|x(\iota_{g(t)-q_{m(t)-1}})\|^2 + \frac{1}{\underline{\alpha}_\sigma} \Biggl [ (\nu_{1_{\sigma}} + \nu_{3_{\sigma}} ) e^{\omega_{2_\sigma}(t-\varphi_{m(t)})} \notag \\
    &  + \mathlarger{\mathlarger{\sum}}_{\substack{l \in \mathbb{N};\\ \iota_{g(t)-q_{m(t)-1}+1} \leq \iota_{g(t)-l} < \iota_{g(t)}}} \Biggl \{ \nu_{1_{\sigma(\iota^+_{g(t)-l})}} \mathlarger{\mathlarger{\prod}}_{\substack{j \in \mathbb{N}_0; \\ \iota_{g(t)-l} < \iota_{g(t)-j} \leq \iota_{g(t)}}} \left ( \frac{\overline{\alpha}_{\sigma(\iota^+_{g(t)-j})}}{\underline{\alpha}_{\sigma(\iota^+_{g(t)-(j+1)})}} \right ) \notag \\
    & e^{- \left(\sum_{i=1}^{n_s} \omega_{1_i} |(\mathcal{O}_i \cap \overline{\Upsilon})(\iota_{g(t)-(l-1)}, t)|\right )} e^{ \left (\sum_{i=1}^{n_s} \omega_{2_i}|(\mathcal{O}_i \cap \overline{\Omega})(\iota_{g(t)-(l-1)}, t)| \right )} \Biggr \} \notag \\
    & +\frac{\overline{\alpha}_{\sigma(\iota^+_{g(t)})} \cdots \overline{\alpha}_{\sigma(\iota^+_{g(t)-q_{m(t)-1}+1})}}{\underline{\alpha}_{\sigma(\iota^+_{g(t)-1})} \cdots \underline{\alpha}_{\sigma(\iota^+_{g(t)-q_{m(t)-1}+1})} \underline{\alpha}_{\sigma(\iota^+_{g(t)-q_{m(t)-1}})}} e^{-\left (\sum_{i=1}^{n_s} \omega_{1_i} | (\mathcal{O}_i \cap \overline{\Upsilon}) (\iota_{g(t)-q_{m(t)-1}+1},t)|\right )} \notag \\
    & e^{\left (\sum_{i=1}^{n_s} \omega_{2_i} |(\mathcal{O}_i \cap \overline{\Omega}) (\iota_{g(t)-q_{m(t)-1}+1},t)| \right)} \Biggl \{ \nu_{1_{\sigma(\iota^+_{g(t)-q_{m(t)-1}})}} \notag \\
    & e^{- \left (\omega_{1_{\sigma(\iota^+_{g(t)-q_{m(t)-1}})}}|\Bar{\Upsilon}(\iota_{g(t)-q_{m(t)-1}},\varphi_{m(t)-1}+v_{m(t)-1})| \right )} \notag \\
    & e^{\left (\omega_{2_{\sigma(\iota^+_{g(t)-q_{m(t)-1}})}}|\Bar{\Omega} (\iota_{g(t)-q_{m(t)-1}},\varphi_{m(t)-1}+v_{m(t)-1})| \right )} + \Biggl (\nu_{1_{\sigma(\iota^+_{g(t)-q_{m(t)-1}})}}  \notag \\
    & +\nu_{3_{\sigma(\iota^+_{g(t)-q_{m(t)-1}})}} \Biggr ) e^{\omega_{2_{\sigma(\iota^+_{g(t)-q_{m(t)-1}})}} v_{m(t)-1}} \Biggr \} \Biggr ] \|w_t\|^2_\infty. \label{eq:trajectory_fsdos_4}
\end{align}
From \eqref{eq:trajectory_fsdos_4}, we can write
\begin{align}
    & \|x(t)\|^2 \leq \frac{\overline{\alpha}_{\sigma(\iota^+_{g(t)})} \cdots \overline{\alpha}_{\sigma(\iota^+_{g(t)-q_{m(t)-1}})}}{\underline{\alpha}_{\sigma(\iota^+_{g(t)})} \cdots \underline{\alpha}_{\sigma(\iota^+_{g(t)-q_{m(t)-1}})}} e^{-\left (\sum_{i=1}^{n_s} \omega_{1_i} | (\mathcal{O}_i \cap \overline{\Upsilon}) (\iota_{g(t)-q_{m(t)-1}},t)|\right )} \notag \\
    & e^{\left (\sum_{i=1}^{n_s} \omega_{2_i} |(\mathcal{O}_i \cap \overline{\Omega}) (\iota_{g(t)-q_{m(t)-1}},t)| \right)} \|x(\iota_{g(t)-q_{m(t)-1}})\|^2 + \frac{1}{\underline{\alpha}_\sigma} \Biggl [ (\nu_{1_{\sigma}} + \nu_{3_{\sigma}} ) e^{\omega_{2_\sigma}(t-\varphi_{m(t)})} \notag \\
    &  + \mathlarger{\mathlarger{\sum}}_{\substack{l \in \mathbb{N};\\ \iota_{g(t)-q_{m(t)-1}} \leq \iota_{g(t)-l} < \iota_{g(t)}}} \Biggl \{ \nu_{1_{\sigma(\iota^+_{g(t)-l})}} \mathlarger{\mathlarger{\prod}}_{\substack{j \in \mathbb{N}_0; \\ \iota_{g(t)-l} < \iota_{g(t)-j} \leq \iota_{g(t)}}} \left ( \frac{\overline{\alpha}_{\sigma(\iota^+_{g(t)-j})}}{\underline{\alpha}_{\sigma(\iota^+_{g(t)-(j+1)})}} \right ) \notag \\
    & e^{- \left(\sum_{i=1}^{n_s} \omega_{1_i} |(\mathcal{O}_i \cap \overline{\Upsilon})(\iota_{g(t)-(l-1)}, t)|\right )} e^{ \left (\sum_{i=1}^{n_s} \omega_{2_i}|(\mathcal{O}_i \cap \overline{\Omega})(\iota_{g(t)-l}, t)| \right )} \Biggr \} + \Biggl (\nu_{1_{\sigma(\iota^+_{g(t)-q_{m(t)-1}})}} \notag \\
    & +\nu_{3_{\sigma(\iota^+_{g(t)-q_{m(t)-1}})}} \Biggr ) \frac{\overline{\alpha}_{\sigma(\iota^+_{g(t)})} \cdots \overline{\alpha}_{\sigma(\iota^+_{g(t)-q_{m(t)-1}+1})}}{\underline{\alpha}_{\sigma(\iota^+_{g(t)-1})} \cdots \underline{\alpha}_{\sigma(\iota^+_{g(t)-q_{m(t)-1}+1})} \underline{\alpha}_{\sigma(\iota^+_{g(t)-q_{m(t)-1}})}} \notag \\
    & e^{-\left (\sum_{i=1}^{n_s} \omega_{1_i} | (\mathcal{O}_i \cap \overline{\Upsilon}) (\iota_{g(t)-q_{m(t)-1}+1},t)|\right )} e^{\left (\sum_{i=1}^{n_s} \omega_{2_i} |(\mathcal{O}_i \cap \overline{\Omega}) (\iota_{g(t)-q_{m(t)-1}},t)| \right)} \Biggr \} \Biggr ] \|w_t\|^2_\infty, \label{eq:trajectory_fsdos_5}
\end{align}
by avoiding $e^{- \left (\omega_{1_{\sigma(\iota^+_{g(t)-q_{m(t)-1}})}}|\Bar{\Upsilon}(\iota_{g(t)-q_{m(t)-1}},\varphi_{m(t)-1}+v_{m(t)-1})| \right )}$ term because $e^{-(a+b)} \leq e^{-a}, \forall \{a,b\} \in \mathbb{R}_{\geq 0} $. Rewrite \eqref{eq:trajectory_fsdos_5}, we get
\begin{align}
    & \|x(t)\|^2 \leq \frac{\overline{\alpha}_{\sigma(\iota^+_{g(t)})} \cdots \overline{\alpha}_{\sigma(\iota^+_{g(t)-q_{m(t)-1}})}}{\underline{\alpha}_{\sigma(\iota^+_{g(t)})} \cdots \underline{\alpha}_{\sigma(\iota^+_{g(t)-q_{m(t)-1}})}} e^{-\left (\sum_{i=1}^{n_s} \omega_{1_i} | (\mathcal{O}_i \cap \overline{\Upsilon}) (\iota_{g(t)-q_{m(t)-1}},t)|\right )} \notag \\
    & e^{\left (\sum_{i=1}^{n_s} \omega_{2_i} |(\mathcal{O}_i \cap \overline{\Omega}) (\iota_{g(t)-q_{m(t)-1}},t)| \right)} \|x(\iota_{g(t)-q_{m(t)-1}})\|^2 + \frac{1}{\underline{\alpha}_{\sigma(\iota^+_{g(t)})}} \Biggl [ \mathlarger{\mathlarger{\sum}}_{\substack{l \in \mathbb{N};\\ \iota_{g(t)-q_{m(t)-1}} \leq \iota_{g(t)-l} < \iota_{g(t)}}} \notag \\
    & \Biggl \{ \nu_{1_{\sigma(\iota^+_{g(t)-l})}} \mathlarger{\mathlarger{\prod}}_{\substack{j \in \mathbb{N}_0; \\ \iota_{g(t)-l} < \iota_{g(t)-j} \leq \iota_{g(t)}}} \left ( \frac{\overline{\alpha}_{\sigma(\iota^+_{g(t)-j})}}{\underline{\alpha}_{\sigma(\iota^+_{g(t)-(j+1)})}} \right )  e^{- \left(\sum_{i=1}^{n_s} \omega_{1_i} |(\mathcal{O}_i \cap \overline{\Upsilon})(\iota_{g(t)-(l-1)}, t)|\right )} \notag \\
    & e^{ \left (\sum_{i=1}^{n_s} \omega_{2_i}|(\mathcal{O}_i \cap \overline{\Omega})(\iota_{g(t)-l}, t)| \right )} \Biggr \} + \mathlarger{\mathlarger{\sum}}_{\substack{k \in \mathbb{N}_0; \\ q \coloneqq \sum_{s=0}^k q_{m(t)-s} \leq q_{m(t)}+q_{m(t)-1}}} \Biggl \{ \Biggl ( \nu_{1_{\sigma(\iota^+_{g(t)-q})}} \notag \\
    & + \nu_{3_{\sigma(\iota^+_{g(t)-q})}} \Biggr )\mathlarger{\prod}_{\substack{{j \in \mathbb{N}_0;}\\{\iota_{g(t)-q} < \iota_{g(t)-j} \leq \iota_{g(t)}}}} \left( \frac{\overline{\alpha}_{\sigma(\iota^+_{g(t)-j})}}{\underline{\alpha}_{\sigma(\iota^+_{g(t)-(j+1)})}} \right ) e^{-(\sum_{i=1}^{n_s} \omega_{1_i}|(\mathcal{O}_i \cap \overline{\Upsilon})(\iota_{g(t)-(q-1)},t)|)} \notag \\ 
    & e^{(\sum_{i=1}^{n_s} \omega_{2_i}|(\mathcal{O}_i \cap \overline{\Omega})(\iota_{g(t)-q},t)|)} \Biggr \} \Biggr ] \|w_t\|_\infty. \label{eq:trajectory_fsdos_6}
\end{align}
Consider $Y_{m(t)-p}$ interval where $q_{m(t)-p}$ jump points are present (see in Figure \ref{fig:time_axis_pth}). Since $x(t)$ is continuous, we can write from \eqref{eq:trajectory_fsdos_14},
\begin{align}
    & \|x(t)\|^2 \leq \frac{\overline{\alpha}_{\sigma(\iota^+_{g(t)})} \cdots \overline{\alpha}_{\sigma(\iota^+_{g(t)-\sum_{j=0}^p q_{m(t)-j}})}}{\underline{\alpha}_{\sigma(\iota^+_{g(t)})} \cdots \underline{\alpha}_{\sigma(\iota^+_{g(t)-\sum_{j=0}^p q_{m(t)-j}})}} e^{-\left (\sum_{i=1}^{n_s} \omega_{1_i} | (\mathcal{O}_i \cap \overline{\Upsilon}) (\iota_{g(t)-\sum_{j=0}^p q_{m(t)-j}},t)|\right )} \notag \\
    & e^{\left (\sum_{i=1}^{n_s} \omega_{2_i} |(\mathcal{O}_i \cap \overline{\Omega}) (\iota_{g(t)-\sum_{j=0}^p q_{m(t)-j}},t)| \right)} \|x(\iota_{g(t)-\sum_{j=0}^p q_{m(t)-j}})\|^2 \notag \\
    & + \frac{1}{\underline{\alpha}_{\sigma(\iota^+_{g(t)})}} \Biggl [ \mathlarger{\mathlarger{\sum}}_{\substack{l \in \mathbb{N};\\ \iota_{g(t)-\sum_{j=0}^p q_{m(t)-j}} \leq \iota_{g(t)-l} < \iota_{g(t)}}} \Biggl \{ \nu_{1_{\sigma(\iota^+_{g(t)-l})}} \mathlarger{\mathlarger{\prod}}_{\substack{j \in \mathbb{N}_0; \\ \iota_{g(t)-l} < \iota_{g(t)-j} \leq \iota_{g(t)}}} \notag \\
    &\left ( \frac{\overline{\alpha}_{\sigma(\iota^+_{g(t)-j})}}{\underline{\alpha}_{\sigma(\iota^+_{g(t)-(j+1)})}} \right )  e^{- \left(\sum_{i=1}^{n_s} \omega_{1_i} |(\mathcal{O}_i \cap \overline{\Upsilon})(\iota_{g(t)-(l-1)}, t)|\right )} e^{ \left (\sum_{i=1}^{n_s} \omega_{2_i}|(\mathcal{O}_i \cap \overline{\Omega})(\iota_{g(t)-l}, t)| \right )} \Biggr \} \notag \\
    & + \mathlarger{\mathlarger{\sum}}_{\substack{k \in \mathbb{N}_0; \\ q \coloneqq \sum_{s=0}^k q_{m(t)-s} \leq \sum_{j=0}^p q_{m(t)-j}}} \Biggl \{ \left ( \nu_{1_{\sigma(\iota^+_{g(t)-q})}} + \nu_{3_{\sigma(\iota^+_{g(t)-q})}} \right )\mathlarger{\prod}_{\substack{{j \in \mathbb{N}_0;}\\{\iota_{g(t)-q} < \iota_{g(t)-j} \leq \iota_{g(t)}}}} \notag \\
    & \left( \frac{\overline{\alpha}_{\sigma(\iota^+_{g(t)-j})}}{\underline{\alpha}_{\sigma(\iota^+_{g(t)-(j+1)})}} \right ) e^{-(\sum_{i=1}^{n_s} \omega_{1_i}|(\mathcal{O}_i \cap \overline{\Upsilon})(\iota_{g(t)-(q-1)},t)|)} e^{(\sum_{i=1}^{n_s} \omega_{2_i}|(\mathcal{O}_i \cap \overline{\Omega})(\iota_{g(t)-q},t)|)} \Biggr \} \Biggr ] \|w_t\|_\infty, \label{eq:trajectory_fsdos_7}
\end{align}
$\forall t \in \mathbb{R}_{\geq \iota_{g(t)-\sum_{j=0}^pq_{m(t)-j}}}$. $Y_{m(t)-(p+1)}$ contains $q_{m(t)-(p+1)}$ jump points. The jump points are $\iota_{g(t)-\sum_{j=0}^{p+1}q_{m(t)-j}+1}$ to $\iota_{g(t)-\sum_{j=0}^{p}q_{m(t)-j}}$. And the last jump points of $Y_{m(t)-(p+2)}$ is  $\iota_{g(t)-\sum_{j=0}^{p+1}q_{m(t)-j}}$ (see in Figure \ref{fig:time_axis_pth}). Hence from \eqref{eq:trajectory_fsdos_1}, we can write
\begin{align}
    &\|x(\iota_{g(t)-\sum_{j=0}^{p}q_{m(t)-j}})\|^2 \leq \frac{\overline{\alpha}_{\sigma(\iota^+_{g(t)-\sum_{j=0}^{p}q_{m(t)-j}-1})} \cdots \overline{\alpha}_{\sigma(\iota^+_{g(t)-\sum_{j=0}^{p+1}q_{m(t)-j}+1})}}{\underline{\alpha}_{\sigma(\iota^+_{g(t)-\sum_{j=0}^{p}q_{m(t)-j}-1})} \cdots \underline{\alpha}_{\sigma(\iota^+_{g(t)-\sum_{j=0}^{p+1}q_{m(t)-j}+1})}} \notag \\
    & e^{-\left (\sum_{i=1}^{n_s} \omega_{1_i} | \mathcal{O}_i (\iota_{g(t)-\sum_{j=0}^{p+1}q_{m(t)-j}+1},\iota_{g(t)-\sum_{j=0}^{p}q_{m(t)-j}})|\right )} \|x(\iota_{g(t)-\sum_{j=0}^{p+1}q_{m(t)-j}+1})\|^2 \notag \\
    & + \Biggl [ \mathlarger{\mathlarger{\sum}}_{\substack{l \in \mathbb{N};\\ \iota_{g(t)-\sum_{j=0}^{p+1}q_{m(t)-j}+1} \leq \iota_{g(t)-l} < \iota_{g(t)-\sum_{j=0}^{p} q_{m(t)-j}}}}\Biggl \{ \frac{\nu_{1_{\sigma(\iota^+_{g(t)-l})}}}  { \underline{\alpha}_{\sigma(\iota^+_{g(t)-\sum_{j=0}^{p} q_{m(t)-j}-1})}} \notag\\ 
    & \prod_{\substack{j \in \mathbb{N}; \\ \iota_{g(t)-l}< \iota_{g(t)-j} < \iota_{g(t)-\sum_{j=0}^{p} q_{m(t)-j}}}} \left ( \frac{ \overline{\alpha}_{\sigma(\iota^+_{g(t)-j})}} { \underline{\alpha}_{\sigma(\iota^+_{g(t)-(j+1)})}} \right ) \notag \\
    & e^{-\left ( \sum_{i=1}^{n_s} \omega_{1_i} |\mathcal{O}_i(\iota_{g(t)-(l-1)}, \iota_{g(t)-\sum_{j=0}^{p}q_{m(t)-j}})| \right) }  \Biggr \} \Biggr ]  {\|w_t\|}^2_\infty. \label{eq:trajectory_fsdos_8}
\end{align}
From \eqref{eq:trajectory_fsdos_3}, we have
\begin{align}
    & \|x(\iota_{g(t)-\sum_{j=0}^{p+1}q_{m(t)-j}+1})\|^2 \leq \frac{\overline{\alpha}_{\sigma(\iota^+_{g(t)-\sum_{j=0}^{p+1}q_{m(t)-j}})}}{\underline{\alpha}_{\sigma(\iota^+_{g(t)-\sum_{j=0}^{p+1}q_{m(t)-j}})}} \notag \\
    & e^{- \left (\omega_{1_{\sigma(\iota^+_{g(t)-\sum_{j=0}^{p+1}q_{m(t)-j}})}}|\Bar{\Upsilon}(\iota_{g(t)-\sum_{j=0}^{p+1}q_{m(t)-j}},\iota_{g(t)-\sum_{j=0}^{p+1}q_{m(t)-j}+1})| \right )} \notag \\
    & e^{\left (\omega_{2_{\sigma(\iota^+_{g(t)-\sum_{j=0}^{p+1}q_{m(t)-j}})}}|\Bar{\Omega} (\iota_{g(t)-\sum_{j=0}^{p+1}q_{m(t)-j}},\iota_{g(t)-\sum_{j=0}^{p+1}q_{m(t)-j}+1})| \right )} \|x(\iota_{g(t)-\sum_{j=0}^{p+1}q_{m(t)-j}})\|^2 \notag \\
    & +\frac{1}{\underline{\alpha}_{\sigma(\iota^+_{g(t)-\sum_{j=0}^{p+1}q_{m(t)-j}})}}\Biggl [ \nu_{1_{\sigma(\iota^+_{g(t)-\sum_{j=0}^{p+1}q_{m(t)-j}})}} \notag \\
    & e^{- \left (\omega_{1_{\sigma(\iota^+_{g(t)-\sum_{j=0}^{p+1}q_{m(t)-j}})}}|\Bar{\Upsilon}(\iota_{g(t)-\sum_{j=0}^{p+1}q_{m(t)-j}},\varphi_{m(t)-(p+1)}+v_{m(t)-(p+1)})| \right )} \notag \\
    & e^{\left (\omega_{2_{\sigma(\iota^+_{g(t)-\sum_{j=0}^{p+1}q_{m(t)-j}})}}|\Bar{\Omega} (\iota_{g(t)-\sum_{j=0}^{p+1}q_{m(t)-j}},\varphi_{m(t)-(p+1)}+v_{m(t)-(p+1)})| \right )} \notag \\
    & + \left (\nu_{1_{\sigma(\iota^+_{g(t)-\sum_{j=0}^{p+1}q_{m(t)-j}})}} +\nu_{3_{\sigma(\iota^+_{g(t)-\sum_{j=0}^{p+1}q_{m(t)-j}})}} \right )e^{\omega_{2_{\sigma(\iota^+_{g(t)-\sum_{j=0}^{p+1}q_{m(t)-j}})}} v_{m(t)-(p+1)}} \Biggr ] \|w_t\|^2_\infty. \label{eq:trajectory_fsdos_9}
\end{align}
Using the inequality of \eqref{eq:trajectory_fsdos_9}, in \eqref{eq:trajectory_fsdos_8}, we get
\begin{align}
    &\|x(\iota_{g(t)-\sum_{j=0}^{p}q_{m(t)-j}})\|^2 \leq \frac{\overline{\alpha}_{\sigma(\iota^+_{g(t)-\sum_{j=0}^{p}q_{m(t)-j}-1})} \cdots \overline{\alpha}_{\sigma(\iota^+_{g(t)-\sum_{j=0}^{p+1}q_{m(t)-j}})}}{\underline{\alpha}_{\sigma(\iota^+_{g(t)-\sum_{j=0}^{p}q_{m(t)-j}-1})} \cdots \underline{\alpha}_{\sigma(\iota^+_{g(t)-\sum_{j=0}^{p+1}q_{m(t)-j}})}} \notag \\
    & e^{-\left (\sum_{i=1}^{n_s} \omega_{1_i} | \mathcal{O}_i (\iota_{g(t)-\sum_{j=0}^{p+1}q_{m(t)-j}},\iota_{g(t)-\sum_{j=0}^{p}q_{m(t)-j}})|\right )} \|x(\iota_{g(t)-\sum_{j=0}^{p+1}q_{m(t)-j}})\|^2 \notag \\
    & + \Biggl [ \mathlarger{\mathlarger{\sum}}_{\substack{l \in \mathbb{N};\\ \iota_{g(t)-\sum_{j=0}^{p+1}q_{m(t)-j}+1} \leq \iota_{g(t)-l} < \iota_{g(t)-\sum_{j=0}^{p} q_{m(t)-j}}}}\Biggl \{ \frac{\nu_{1_{\sigma(\iota^+_{g(t)-l})}}}  { \underline{\alpha}_{\sigma(\iota^+_{g(t)-\sum_{j=0}^{p} q_{m(t)-j}-1})}} \notag\\ 
    & \prod_{\substack{j \in \mathbb{N}; \\ \iota_{g(t)-l}< \iota_{g(t)-j} < \iota_{g(t)-\sum_{j=0}^{p} q_{m(t)-j}}}} \Bigl ( \frac{ \overline{\alpha}_{\sigma(\iota^+_{g(t)-j})}} { \underline{\alpha}_{\sigma(\iota^+_{g(t)-(j+1)})}} \Bigr ) e^{-\left ( \sum_{i=1}^{n_s} \omega_{1_i} |\mathcal{O}_i(\iota_{g(t)-(l-1)}, \iota_{g(t)-\sum_{j=0}^{p}q_{m(t)-j}})| \right) }  \Biggr \}  \notag \\
    & + \frac{\overline{\alpha}_{\sigma(\iota^+_{g(t)-\sum_{j=0}^{p}q_{m(t)-j}-1})} \cdots \overline{\alpha}_{\sigma(\iota^+_{g(t)-\sum_{j=0}^{p+1}q_{m(t)-j}+1})}}{\underline{\alpha}_{\sigma(\iota^+_{g(t)-\sum_{j=0}^{p}q_{m(t)-j}-1})} \cdots \underline{\alpha}_{\sigma(\iota^+_{g(t)-\sum_{j=0}^{p+1}q_{m(t)-j}+1})} \underline{\alpha}_{\sigma(\iota^+_{g(t)-\sum_{j=0}^{p+1}q_{m(t)-j}})}} \notag \\
    & e^{-\left (\sum_{i=1}^{n_s} \omega_{1_i} | \mathcal{O}_i (\iota_{g(t)-\sum_{j=0}^{p+1}q_{m(t)-j}+1},\iota_{g(t)-\sum_{j=0}^{p}q_{m(t)-j}})|\right )} \Biggl \{ \nu_{1_{\sigma(\iota^+_{g(t)-\sum_{j=0}^{p+1}q_{m(t)-j}})}} \notag \\
    & e^{- \left (\omega_{1_{\sigma(\iota^+_{g(t)-\sum_{j=0}^{p+1}q_{m(t)-j}})}}|\Bar{\Upsilon}(\iota_{g(t)-\sum_{j=0}^{p+1}q_{m(t)-j}},\varphi_{m(t)-(p+1)}+v_{m(t)-(p+1)})| \right )} \notag \\
    & e^{\left (\omega_{2_{\sigma(\iota^+_{g(t)-\sum_{j=0}^{p+1}q_{m(t)-j}})}}|\Bar{\Omega} (\iota_{g(t)-\sum_{j=0}^{p+1}q_{m(t)-j}},\varphi_{m(t)-(p+1)}+v_{m(t)-(p+1)})| \right )} \notag \\
    & + \left (\nu_{1_{\sigma(\iota^+_{g(t)-\sum_{j=0}^{p+1}q_{m(t)-j}})}} +\nu_{3_{\sigma(\iota^+_{g(t)-\sum_{j=0}^{p+1}q_{m(t)-j}})}} \right )e^{\omega_{2_{\sigma(\iota^+_{g(t)-\sum_{j=0}^{p+1}q_{m(t)-j}})}} v_{m(t)-(p+1)}} \Biggr \} \Biggr ] \notag \\
    & \|w_t\|^2_\infty. \label{eq:trajectory_fsdos_10}
\end{align}
By resizing \eqref{eq:trajectory_fsdos_10}, we get
\begin{align}
    &\|x(\iota_{g(t)-\sum_{j=0}^{p}q_{m(t)-j}})\|^2 \leq \frac{\overline{\alpha}_{\sigma(\iota^+_{g(t)-\sum_{j=0}^{p}q_{m(t)-j}-1})} \cdots \overline{\alpha}_{\sigma(\iota^+_{g(t)-\sum_{j=0}^{p+1}q_{m(t)-j}})}}{\underline{\alpha}_{\sigma(\iota^+_{g(t)-\sum_{j=0}^{p}q_{m(t)-j}-1})} \cdots \underline{\alpha}_{\sigma(\iota^+_{g(t)-\sum_{j=0}^{p+1}q_{m(t)-j}})}} \notag \\
    & e^{-\left (\sum_{i=1}^{n_s} \omega_{1_i} | \mathcal{O}_i (\iota_{g(t)-\sum_{j=0}^{p+1}q_{m(t)-j}},\iota_{g(t)-\sum_{j=0}^{p}q_{m(t)-j}})|\right )} \|x(\iota_{g(t)-\sum_{j=0}^{p+1}q_{m(t)-j}})\|^2 \notag \\
    & + \frac{1}{\underline{\alpha}_{\sigma(\iota^+_{g(t)-\sum_{j=0}^{p}q_{m(t)-j}-1})}} \Biggl [ \mathlarger{\mathlarger{\sum}}_{\substack{l \in \mathbb{N};\\ \iota_{g(t)-\sum_{j=0}^{p+1}q_{m(t)-j}} \leq \iota_{g(t)-l} < \iota_{g(t)-\sum_{j=0}^{p} q_{m(t)-j}}}} \notag \\
    & \Biggl \{ \nu_{1_{\sigma(\iota^+_{g(t)-l})}} \prod_{\substack{j \in \mathbb{N}; \\ \iota_{g(t)-l}< \iota_{g(t)-j} < \iota_{g(t)-\sum_{j=0}^{p} q_{m(t)-j}}}} \left ( \frac{ \overline{\alpha}_{\sigma(\iota^+_{g(t)-j})}} { \underline{\alpha}_{\sigma(\iota^+_{g(t)-(j+1)})}} \right ) \notag \\
    & e^{-\left ( \sum_{i=1}^{n_s} \omega_{1_i} |\mathcal{O}_i(\iota_{g(t)-(l-1)}, \iota_{g(t)-\sum_{j=0}^{p}q_{m(t)-j}})| \right) }  \Biggr \}  \notag \\
    & e^{\left (\omega_{2_{\sigma(\iota^+_{g(t)-\sum_{j=0}^{p+1}q_{m(t)-j}})}}|\Bar{\Omega} (\iota_{g(t)-\sum_{j=0}^{p+1}q_{m(t)-j}},\varphi_{m(t)-(p+1)}+v_{m(t)-(p+1)})| \right )} \notag \\
    & + \left (\nu_{1_{\sigma(\iota^+_{g(t)-\sum_{j=0}^{p+1}q_{m(t)-j}})}} +\nu_{3_{\sigma(\iota^+_{g(t)-\sum_{j=0}^{p+1}q_{m(t)-j}})}} \right )\notag \\
    & \frac{\overline{\alpha}_{\sigma(\iota^+_{g(t)-\sum_{j=0}^{p}q_{m(t)-j}-1})} \cdots \overline{\alpha}_{\sigma(\iota^+_{g(t)-\sum_{j=0}^{p+1}q_{m(t)-j}+1})}}{\underline{\alpha}_{\sigma(\iota^+_{g(t)-\sum_{j=0}^{p}q_{m(t)-j}-2})} \cdots \underline{\alpha}_{\sigma(\iota^+_{g(t)-\sum_{j=0}^{p+1}q_{m(t)-j}})}} e^{\omega_{2_{\sigma(\iota^+_{g(t)-\sum_{j=0}^{p+1}q_{m(t)-j}})}} v_{m(t)-(p+1)}} \Biggr \} \Biggr ] \notag \\
    & \|w_t\|^2_\infty, \label{eq:trajectory_fsdos_11}
\end{align}
by avoiding $e^{- \left (\omega_{1_{\sigma(\iota^+_{g(t)-\sum_{j=0}^{p+1}q_{m(t)-j}})}}|\Bar{\Upsilon}(\iota_{g(t)-\sum_{j=0}^{p+1}q_{m(t)-j}},\varphi_{m(t)-(p+1)}+v_{m(t)-(p+1)})| \right )}$ term because $e^{-(a+b)} \leq e^{-a}, \forall \{a,b\} \in \mathbb{R}_{\geq 0}$. Therefor using the inequality \eqref{eq:trajectory_fsdos_11} in \eqref{eq:trajectory_fsdos_7}, we have
\begin{align}
    & \|x(t)\|^2 \leq \frac{\overline{\alpha}_{\sigma(\iota^+_{g(t)})} \cdots \overline{\alpha}_{\sigma(\iota^+_{g(t)-\sum_{j=0}^{p+1} q_{m(t)-j}})}}{\underline{\alpha}_{\sigma(\iota^+_{g(t)})} \cdots \underline{\alpha}_{\sigma(\iota^+_{g(t)-\sum_{j=0}^{p+1} q_{m(t)-j}})}} e^{-\left (\sum_{i=1}^{n_s} \omega_{1_i} | (\mathcal{O}_i \cap \overline{\Upsilon}) (\iota_{g(t)-\sum_{j=0}^{p+1} q_{m(t)-j}},t)|\right )} \notag \\
    & e^{\left (\sum_{i=1}^{n_s} \omega_{2_i} |(\mathcal{O}_i \cap \overline{\Omega}) (\iota_{g(t)-\sum_{j=0}^{p+1} q_{m(t)-j}},t)| \right)} \|x(\iota_{g(t)-\sum_{j=0}^{p+1} q_{m(t)-j}})\|^2 \notag \\
    & + \frac{1}{\underline{\alpha}_{\sigma(\iota^+_{g(t)})}} \Biggl [ \mathlarger{\mathlarger{\sum}}_{\substack{l \in \mathbb{N};\\ \iota_{g(t)-\sum_{j=0}^p q_{m(t)-j}} \leq \iota_{g(t)-l} < \iota_{g(t)}}} \Biggl \{ \nu_{1_{\sigma(\iota^+_{g(t)-l})}} \mathlarger{\mathlarger{\prod}}_{\substack{j \in \mathbb{N}_0; \\ \iota_{g(t)-l} < \iota_{g(t)-j} \leq \iota_{g(t)}}} \notag \\
    &\left ( \frac{\overline{\alpha}_{\sigma(\iota^+_{g(t)-j})}}{\underline{\alpha}_{\sigma(\iota^+_{g(t)-(j+1)})}} \right )  e^{- \left(\sum_{i=1}^{n_s} \omega_{1_i} |(\mathcal{O}_i \cap \overline{\Upsilon})(\iota_{g(t)-(l-1)}, t)|\right )} e^{ \left (\sum_{i=1}^{n_s} \omega_{2_i}|(\mathcal{O}_i \cap \overline{\Omega})(\iota_{g(t)-l}, t)| \right )} \Biggr \} \notag \\
    & + \mathlarger{\mathlarger{\sum}}_{\substack{k \in \mathbb{N}_0; \\ q \coloneqq \sum_{s=0}^k q_{m(t)-s} \leq \sum_{j=0}^p q_{m(t)-j}}} \Biggl \{ \left ( \nu_{1_{\sigma(\iota^+_{g(t)-q})}} + \nu_{3_{\sigma(\iota^+_{g(t)-q})}} \right )\mathlarger{\prod}_{\substack{{j \in \mathbb{N}_0;}\\{\iota_{g(t)-q} < \iota_{g(t)-j} \leq \iota_{g(t)}}}} \notag \\
    & \left( \frac{\overline{\alpha}_{\sigma(\iota^+_{g(t)-j})}}{\underline{\alpha}_{\sigma(\iota^+_{g(t)-(j+1)})}} \right ) e^{-(\sum_{i=1}^{n_s} \omega_{1_i}|(\mathcal{O}_i \cap \overline{\Upsilon})(\iota_{g(t)-(q-1)},t)|)} e^{(\sum_{i=1}^{n_s} \omega_{2_i}|(\mathcal{O}_i \cap \overline{\Omega})(\iota_{g(t)-q},t)|)} \Biggr \} \notag \\
    & +\frac{\overline{\alpha}_{\sigma(\iota^+_{g(t)})} \cdots \overline{\alpha}_{\sigma(\iota^+_{g(t)-\sum_{j=0}^p q_{m(t)-j}})}}{\underline{\alpha}_{\sigma(\iota^+_{g(t)-1})} \cdots \underline{\alpha}_{\sigma(\iota^+_{g(t)-\sum_{j=0}^p q_{m(t)-j}-1})}} e^{-\left (\sum_{i=1}^{n_s} \omega_{1_i} | (\mathcal{O}_i \cap \overline{\Upsilon}) (\iota_{g(t)-\sum_{j=0}^p q_{m(t)-j}},t)|\right )} \notag \\
    & e^{\left (\sum_{i=1}^{n_s} \omega_{2_i} |(\mathcal{O}_i \cap \overline{\Omega}) (\iota_{g(t)-\sum_{j=0}^p q_{m(t)-j}},t)| \right)} \Biggl [ \mathlarger{\mathlarger{\sum}}_{\substack{l \in \mathbb{N};\\ \iota_{g(t)-\sum_{j=0}^{p+1}q_{m(t)-j}} \leq \iota_{g(t)-l} < \iota_{g(t)-\sum_{j=0}^{p} q_{m(t)-j}}}} \notag \\
    & \Biggl \{ \nu_{1_{\sigma(\iota^+_{g(t)-l})}} \prod_{\substack{j \in \mathbb{N}; \\ \iota_{g(t)-l}< \iota_{g(t)-j} < \iota_{g(t)-\sum_{j=0}^{p} q_{m(t)-j}}}} \left ( \frac{ \overline{\alpha}_{\sigma(\iota^+_{g(t)-j})}} { \underline{\alpha}_{\sigma(\iota^+_{g(t)-(j+1)})}} \right ) \notag \\
    & e^{-\left ( \sum_{i=1}^{n_s} \omega_{1_i} |\mathcal{O}_i(\iota_{g(t)-(l-1)}, \iota_{g(t)-\sum_{j=0}^{p}q_{m(t)-j}})| \right) }  \Biggr \}  \notag \\
    & e^{\left (\omega_{2_{\sigma(\iota^+_{g(t)-\sum_{j=0}^{p+1}q_{m(t)-j}})}}|\Bar{\Omega} (\iota_{g(t)-\sum_{j=0}^{p+1}q_{m(t)-j}},\varphi_{m(t)-(p+1)}+v_{m(t)-(p+1)})| \right )} \notag
\end{align}
\begin{align}
    & + \left (\nu_{1_{\sigma(\iota^+_{g(t)-\sum_{j=0}^{p+1}q_{m(t)-j}})}} +\nu_{3_{\sigma(\iota^+_{g(t)-\sum_{j=0}^{p+1}q_{m(t)-j}})}} \right )\notag \\
    & \frac{\overline{\alpha}_{\sigma(\iota^+_{g(t)-\sum_{j=0}^{p}q_{m(t)-j}-1})} \cdots \overline{\alpha}_{\sigma(\iota^+_{g(t)-\sum_{j=0}^{p+1}q_{m(t)-j}+1})}}{\underline{\alpha}_{\sigma(\iota^+_{g(t)-\sum_{j=0}^{p}q_{m(t)-j}-2})} \cdots \underline{\alpha}_{\sigma(\iota^+_{g(t)-\sum_{j=0}^{p+1}q_{m(t)-j}})}} e^{\omega_{2_{\sigma(\iota^+_{g(t)-\sum_{j=0}^{p+1}q_{m(t)-j}})}} v_{m(t)-(p+1)}} \Biggr \} \Biggr ] \Biggr ] \notag \\
    & \|w_t\|^2_\infty. \label{eq:trajectory_fsdos_12}
\end{align}
Resizing \eqref{eq:trajectory_fsdos_12}, we get
\begin{align}
    & \|x(t)\|^2 \leq \frac{\overline{\alpha}_{\sigma(\iota^+_{g(t)})} \cdots \overline{\alpha}_{\sigma(\iota^+_{g(t)-\sum_{j=0}^{p+1} q_{m(t)-j}})}}{\underline{\alpha}_{\sigma(\iota^+_{g(t)})} \cdots \underline{\alpha}_{\sigma(\iota^+_{g(t)-\sum_{j=0}^{p+1} q_{m(t)-j}})}} e^{-\left (\sum_{i=1}^{n_s} \omega_{1_i} | (\mathcal{O}_i \cap \overline{\Upsilon}) (\iota_{g(t)-\sum_{j=0}^{p+1} q_{m(t)-j}},t)|\right )} \notag \\
    & e^{\left (\sum_{i=1}^{n_s} \omega_{2_i} |(\mathcal{O}_i \cap \overline{\Omega}) (\iota_{g(t)-\sum_{j=0}^{p+1} q_{m(t)-j}},t)| \right)} \|x(\iota_{g(t)-\sum_{j=0}^{p+1} q_{m(t)-j}})\|^2 \notag \\
    & + \frac{1}{\underline{\alpha}_{\sigma(\iota^+_{g(t)})}} \Biggl [ \mathlarger{\mathlarger{\sum}}_{\substack{l \in \mathbb{N};\\ \iota_{g(t)-\sum_{j=0}^{p+1} q_{m(t)-j}} \leq \iota_{g(t)-l} < \iota_{g(t)}}} \Biggl \{ \nu_{1_{\sigma(\iota^+_{g(t)-l})}} \mathlarger{\mathlarger{\prod}}_{\substack{j \in \mathbb{N}_0; \\ \iota_{g(t)-l} < \iota_{g(t)-j} \leq \iota_{g(t)}}} \notag \\
    &\left ( \frac{\overline{\alpha}_{\sigma(\iota^+_{g(t)-j})}}{\underline{\alpha}_{\sigma(\iota^+_{g(t)-(j+1)})}} \right )  e^{- \left(\sum_{i=1}^{n_s} \omega_{1_i} |(\mathcal{O}_i \cap \overline{\Upsilon})(\iota_{g(t)-(l-1)}, t)|\right )} e^{ \left (\sum_{i=1}^{n_s} \omega_{2_i}|(\mathcal{O}_i \cap \overline{\Omega})(\iota_{g(t)-l}, t)| \right )} \Biggr \} \notag \\
    & + \mathlarger{\mathlarger{\sum}}_{\substack{k \in \mathbb{N}_0; \\ q \coloneqq \sum_{s=0}^k q_{m(t)-s} \leq \sum_{j=0}^{p+1} q_{m(t)-j}}} \Biggl \{ \left ( \nu_{1_{\sigma(\iota^+_{g(t)-q})}} + \nu_{3_{\sigma(\iota^+_{g(t)-q})}} \right )\mathlarger{\prod}_{\substack{{j \in \mathbb{N}_0;}\\{\iota_{g(t)-q} < \iota_{g(t)-j} \leq \iota_{g(t)}}}} \notag \\
    & \left( \frac{\overline{\alpha}_{\sigma(\iota^+_{g(t)-j})}}{\underline{\alpha}_{\sigma(\iota^+_{g(t)-(j+1)})}} \right ) e^{-(\sum_{i=1}^{n_s} \omega_{1_i}|(\mathcal{O}_i \cap \overline{\Upsilon})(\iota_{g(t)-(q-1)},t)|)} e^{(\sum_{i=1}^{n_s} \omega_{2_i}|(\mathcal{O}_i \cap \overline{\Omega})(\iota_{g(t)-q},t)|)} \Biggr \} \Biggr ] \|w_t\|_\infty, \label{eq:trajectory_fsdos_13}
\end{align}
$\forall t \in \mathbb{R}_{\geq \iota_{g(t)-\sum_{j=0}^{p+1} q_{m(t)-j}}},$ which concludes the proof.
We must bind the sum terms in \eqref{eq:trajectory_fsdos_14} in order to complete the proof. In the next lemma, we demonstrate that.
\begin{lemma} \label{lm:upper_bound_sum}
    Under Assumptions \ref{as1}, \ref{as2}, and \ref{as3} and condition \eqref{eq:convergance_1} and \eqref{eq:varkappa} in Lemma \ref{lm:convergence}, the sum
    \begin{align}
            & \frac{1}{\underline{\alpha}_{\sigma(\iota^+_{g(t)})}} \mathlarger{\sum}_{\substack{{k \in \mathbb{N}_0;}\\{q\coloneqq   \sum_{s=0}^k q_{m(t)-s} \leq g(t)}}}  \left(\nu_{1_{\sigma(\iota^+_{g(t)-q})}} + \nu_{3_{\sigma(\iota^+_{g(t)-q})}} \right ) \mathlarger{\prod}_{\substack{{j \in \mathbb{N}_0;}\\{\iota_{g(t)-q} < \iota_{g(t)-j} \leq \iota_{g(t)}}}}  \notag \\
            & \left( \frac{\overline{\alpha}_{\sigma(\iota^+_{g(t)-j})}}{\underline{\alpha}_{\sigma(\iota^+_{g(t)-(j+1)})}} \right ) e^{-(\sum_{i=1}^{n_s} \omega_{1_i}|(\mathcal{O}_i \cap \overline{\Upsilon})(\iota_{g(t)-(q-1)},t)|)} e^{(\sum_{i=1}^{n_s} \omega_{2_i}|(\mathcal{O}_i \cap \overline{\Omega})(\iota_{g(t)-q},t)|)} \label{eq:sum_fsdos_1}
    \end{align}
    is bounded from above.
\end{lemma}
$Proof$ $of$ $Lemma$ $\ref{lm:upper_bound_sum}$ $:$ Let's take the exponential term of \eqref{eq:sum_fsdos_1},
\begin{align}
    & \sum_{\substack{{k \in \mathbb{N}_0;}\\{q \leq g(t)}}} e^{-(\sum_{i=1}^{n_s} \omega_{1_i}|(\mathcal{O}_i \cap \overline{\Upsilon})(\iota_{g(t)-(q-1)},t)|)} e^{(\sum_{i=1}^{n_s} \omega_{2_i}|(\mathcal{O}_i \cap \overline{\Omega})(\iota_{g(t)-q},t)|)} \notag \\
    & = \sum_{\substack{{k \in \mathbb{N}_0;}\\{q \leq g(t)}}} \exp{\left \{ \substack{-\bigl [\omega_{1_{\sigma(\iota^+_{g(t)})}} |\overline{\Upsilon}(\iota_{g(t)},t)| + \omega_{1_{\sigma(\iota^+_{g(t)-1})}} |\overline{\Upsilon} (\iota_{g(t)-1}, \iota_{g(t)})| + \cdots \\ +\omega_{1_{\sigma(\iota^+_{g(t)-(q-1)})}} |\overline{\Upsilon} (\iota_{g(t)-(q-1)}, \iota_{g(t)-(q-2)})| \bigr ]}\right \}} \notag \\
    & \exp{\left \{ \substack{\omega_{2_{\sigma(\iota^+_{g(t)})}} |\overline{\Omega}(\varphi_{m(t)},t) | + \omega_{2_{\sigma(\iota^+_{g(t)-q_{m(t)-1}})}} | \overline{\Omega}(\varphi_{m(t)-1}, \iota_{g(t)-q_{m(t)-1}+1}|) \\ + \cdots + \omega_{2_{\sigma (\iota^+_{g(t)-q})}} |\overline{\Omega}(\varphi_{m(t)-k}, \iota_{g(t)-(q-1)})| }\right \}}. \label{eq:sum_fsdos_2}
\end{align}
Note that
\begin{align}
    |\overline{\Omega}(\tau,t)| \leq |\Omega(\tau,t)|+(1+ n(\tau,t)) \Delta_*
\end{align}
$\forall t, \tau \in \mathbb{R}_{\geq 0}$ with $t> \tau$. In words, $|\overline{\Omega}(\tau,t)|$ can be upper bounded by the total length of DoS over $[\tau, t]$ plus the maximum actuation delay $\Delta_*$, which may occur once at the beginning of the interval $[\tau, t]$ (as a result of a previous FSDoS) plus $n(\tau, t)$ times, where $n(\tau, t)$ represents the number of off/on transitions of FSDoS occurring over $[\tau, t)$. From Assumption \ref{as2} and \ref{as3}, we can write
\begin{align}
    |\overline{\Omega}(\tau,t)| &\leq \zeta +\frac{t-\tau}{T} + \left ( 1+ \eta \frac{t-\tau}{\tau_f}\right ) \Delta_* \notag \\
    & \eqqcolon \zeta_* +\frac{t-\tau}{T_*},
\end{align}
considering $\zeta_* \coloneqq \zeta +(1+\eta) \Delta_*$ and $T_* \coloneqq T \tau_F/ (T \Delta_*+ \tau_F).$ With this in mind, we can now analyze the sum term in \eqref{eq:sum_fsdos_2}. From the above inequality, we have
\begin{align}
    |\overline{\Omega}(\varphi_{m(t)},t)| \leq \zeta_* + \frac{t-\varphi_{m(t)}}{T_*}
\end{align}
$\forall t \in \mathbb{R}_{\geq \varphi_{m(t)}}, m \in \mathbb{N}$. Similarly, $|\overline{\Omega}(\varphi_{m(t)-1}, \iota_{g(t)-(q_{m(t)-1}-1)})| \leq (\iota_{g(t)-(q_{m(t)-1}-1)}- \varphi_{m(t)-1})/ T_*$ and $|\overline{\Omega}(\varphi_{m(t)-k}, \iota_{g(t)-(q-1)})| \leq (\iota_{g(t)-(q-1)}-\varphi_{m(t)-k})/T_*$. Consider next $|\overline{\Upsilon}(\iota_{g(t)-(q-1)}, \iota_{g(t)-(q-2)})|$. We are avoiding the constant terms $\zeta_*$ in the rest of the intervals. We have
\begin{align}
    |\overline{\Upsilon}(\iota_{g(t)},t)| =t- \iota_{g(t)} - |\overline{\Omega}(\varphi_{m(t)},t)|
\end{align}
$\forall t \in \mathbb{R}_{\varphi_{m(t)}},$ where $\varphi_{m(t)}> \iota_{g(t)}$. Similarly, $|\overline{\Upsilon}(\iota_{g(t)-1}, \iota_{g(t)})| = \iota_{g(t)}- \iota_{g(t)-1}-|\overline{\Omega}(\iota_{g(t)-1}, \iota_{g(t)})|= \iota_{g(t)}- \iota_{g(t)-1}, |\overline{\Upsilon}(\iota_{g(t)-q_{m(t)-1}}, \iota_{g(t)-q_{m(t)-1}+1})| = \iota_{g(t)-q_{m(t)-1}+1}- \iota_{g(t)-q_{m(t)-1}}- |\overline{\Omega}(\varphi_{m(t)-1},\iota_{g(t)-q_{m(t)-1}+1})|,$ and $|\overline{\Upsilon}(\iota_{g(t)-(q-1)}, \iota_{g(t)-(q-2)})|= \iota_{g(t)-(q-2)}- \varphi_{m(t)-k} - |\overline{\Omega}(\varphi_{m(t)-k}, \iota_{g(t)-(q-1)})|$ Hence, from \eqref{eq:sum_fsdos_2}, we can write
\begin{align}
    \sum_{\substack{k \in \mathbb{N}_0; \\ q \leq g(t)}} & \exp{\left \{ \substack{-\Biggl [\sum_{i=1}^{n_s} \omega_{1_i} |\mathcal{O}_i (\varphi_{m(t)-k},t)|- \Bigl\{ \left ( \omega_{1_{\sigma(\iota^+_{g(t)})}} + \omega_{2_{\sigma(\iota^+_{g(t)}) }} \right ) |\overline{\Omega}(\varphi_{m(t)},t)| \\ + \left ( \omega_{1_{\sigma(\iota^+_{g(t)-q_{m(t)-1}})}} + \omega_{2_{\sigma(\iota^+_{g(t)-q_{m(t)-1}}) }} \right ) | \overline{\Omega}(\varphi_{m(t)-1}, \iota_{g(t)-(q_{m(t)-1}-1)})| \\+ \cdots + \left ( \omega_{1_{\sigma(\iota^+_{g(t)-q})}} + \omega_{2_{\sigma(\iota^+_{g(t)-q})}} \right ) |\overline{\Omega} (\varphi_{m(t)-k}, \iota_{g(t)-(q-1)})| \Bigr \} \Biggr ]} \right \}} \notag \\
    \leq & \exp{\left\{ \left ( \omega_{1_{\sigma(\iota^+_{g(t)})}} + \omega_{2_{\sigma(\iota^+_{g(t)})}}\right ) \zeta_*\right \}} \notag \\
    & \sum_{\substack{k \in \mathbb{N}_0; \\ q \leq g(t)}} \exp{\left \{ \substack{- \Bigl[ \sum_{i=1}^{n_s} \omega_{1_i} |\mathcal{O}_i(\varphi_{m(t)-k},t)| - \frac{1}{T_*} \Bigl \{ \left ( \omega_{1_{\sigma(\iota^+_{g(t)})}} + \omega_{2_{\sigma(\iota^+_{g(t)})}}\right ) (t-\varphi_{m(t)})\\ + \left ( \omega_{1_{\sigma(\iota^+_{g(t)-q_{m(t)-1}})}} + \omega_{2_{\sigma(\iota^+_{g(t)-q_{m(t)-1}})}}\right ) (\iota_{g(t)-(q_{m(t)-1}-1)}- \varphi_{m(t)-1}) \\ + \cdots  + \left ( \omega_{1_{\sigma(\iota^+_{g(t)-q})}} + \omega_{2_{\sigma(\iota^+_{g(t)-q})}} \right ) \left ( \iota_{g(t)-(q-1)}- \varphi_{m(t)-k} \right ) \Bigr ]} \right\}}. \label{eq:sum_fsdos_3}
\end{align}
Moreover, Assumption \ref{as1} and Assumption \ref{as2} yields
\begin{align}
    t- \varphi_{m(t)} \geq \tau_F n(\varphi_{m(t)},t) - \tau_F \eta= \tau_F- \tau_F \eta \label{eq:assumption_fsdos_1}
\end{align}
where $n(\varphi_{m(t)},t)=1$ because only one FSDoS on transition is present in $[ \varphi_{m(t)}, t)$ interval and $l(\varphi_{m(t)}, t)=0$ because $[\varphi_{m(t)},t)$ is operated in a single switching mode considering that $t< \varphi_{m(t)}+v_{m(t)}$. Similarly, we get $\iota_{g(t)-(q_{m(t)-1}-1)}- \varphi_{m(t)-1} \geq \tau_F, \iota_{g(t)-(q-1)}- \varphi_{m(t)-k} \geq \tau_F$ and so on by avoiding the constant term $\eta$ in the rest of the cases. Again, Assumption \ref{as1} and Assumption \ref{as2} yields
\begin{align}
    \varphi_{m(t)}- \iota_{g(t)} \geq \tau_D l(\iota_{g(t)}, \varphi_{m(t)})- \tau_D \varkappa = \tau_D - \tau_D \varkappa
\end{align}
where $l(\iota_{g(t)}, \varphi_{m(t)})=1$ because only one time switching is changed in $[\iota_{g(t)}, \varphi_{m(t)})$ interval and $n(\iota_{g(t)}, \varphi_{m(t)})=0$ because there is no FSDoS present in the interval $[\iota_{g(t)}, \varphi_{m(t)})$. Similarly, we get $\iota_{g(t)}- \iota_{g(t)-1} \geq \tau_D, \varphi_{m(t)-1}- \iota_{g(t)-(q_{m(t)-1})} \geq \tau_D, \iota_{g(t)-q}- \iota_{g(t)-(q+1)} \geq \tau_D $ and so on by avoiding the constant term $\nu$. Hence
\begin{align}
    & \sum_{\substack{k \in \mathbb{N}_0; \\ q \leq g(t)}} \exp{\left \{ \substack{- \Bigl[ \sum_{i=1}^{n_s} \omega_{1_i} |\mathcal{O}_i(\varphi_{m(t)-k},t)| - \frac{1}{T_*} \Bigl \{ \left ( \omega_{1_{\sigma(\iota^+_{g(t)})}} + \omega_{2_{\sigma(\iota^+_{g(t)})}}\right ) (t-\varphi_{m(t)})\\ + \left ( \omega_{1_{\sigma(\iota^+_{g(t)-q_{m(t)-1}})}} + \omega_{2_{\sigma(\iota^+_{g(t)-q_{m(t)-1}})}}\right ) (\iota_{g(t)-(q_{m(t)-1}-1)}- \varphi_{m(t)-1}) \\ + \cdots  + \left ( \omega_{1_{\sigma(\iota^+_{g(t)-q})}} + \omega_{2_{\sigma(\iota^+_{g(t)-q})}} \right ) \left ( \iota_{g(t)-(q-1)}- \varphi_{m(t)-k} \right ) \Bigr \} \Bigr ]} \right\}} \notag \\
    & =\sum_{\substack{k \in \mathbb{N}_0; \\ q \leq g(t)}} \exp{\left \{ \substack{ -\Bigl [ \omega_{1_{\sigma (\iota^+_{g(t)})}} \left (t- \iota_{g(t)} \right ) + \omega_{1_{\sigma (\iota^+_{g(t)-1})}} \left (\iota_{g(t)}- \iota_{g(t)-1} \right ) + \cdots + \omega_{1_{\sigma(\iota^+_{g(t)-q_{m(t)-1}})}}\\ \left (\iota_{g(t)-(q_{m(t)-1}-1)}- \iota_{g(t)-q_{m(t)-1}}\right) + \cdots + \omega_{1_{\sigma(\iota^+_{g(t)-q})}} \left (\iota_{g(t)-(q-1)}- \varphi_{m(t)-k} \right) \\- \frac{1}{T_*} \Bigl \{ \left ( \omega_{1_{\sigma(\iota^+_{g(t)})}} + \omega_{2_{\sigma(\iota^+_{g(t)})}}\right ) (t-\varphi_{m(t)}) + \Bigl ( \omega_{1_{\sigma(\iota^+_{g(t)-q_{m(t)-1}})}} \\ + \omega_{2_{\sigma(\iota^+_{g(t)-q_{m(t)-1}})}} \Bigr ) \left (\iota_{g(t)-(q_{m(t)-1}-1)}- \varphi_{m(t)-1} \right) + \cdots \\ + \left ( \omega_{1_{\sigma(\iota^+_{g(t)-q})}} + \omega_{2_{\sigma(\iota^+_{g(t)-q})}} \right ) \left ( \iota_{g(t)-(q-1)}- \varphi_{m(t)-k} \right ) \Bigr \}\Bigr ]}\right \}} \notag \\
    & \leq \exp{ \left \{\omega_{1_{\sigma(\iota^+_{g(t)})}} \tau_D \varkappa \right \}} \sum_{\substack{k \in \mathbb{N}_0; \\ q \leq g(t)}} \exp{ \left\{ - \left ( \omega_{1_{\sigma(\iota^+_{g(t)})}} +\cdots + \omega_{1_{\sigma(\iota^+_{g(t)-q})}}\right ) \tau_D \right \}} \notag \\
    & \exp{ \left \{ \substack{ -\Biggl[ \left ( \omega_{1_{\sigma(\iota^+_{g(t)})}} - \frac{\omega_{1_{\sigma(\iota^+_{g(t)})}}+ \omega_{2_{\sigma(\iota^+_{g(t)})}}}{T_*}\right )(t- \varphi_{m(t)}) + \Biggl ( \omega_{1_{\sigma(\iota^+_{g(t)-q_{m(t)-1}})}} \\ -\frac{\omega_{1_{\sigma(\iota^+_{g(t)-q_{m(t)-1}})}}+ \omega_{2_{\sigma(\iota^+_{g(t)-q_{m(t)-1}})}}}{T_*} \Biggr )\left ( \iota_{g(t)-(q_{m(t)-1}-1)}- \varphi_{m(t)-1} \right ) + \cdots + \\ \left ( \omega_{1_{\sigma(\iota^+_{g(t)-q})}} - \frac{\omega_{1_{\sigma(\iota^+_{g(t)-q})}} + \omega_{1_{\sigma(\iota^+_{g(t)-q})}}}{T_*}\right ) \left ( \iota_{g(t)-(q-1)}- \varphi_{m(t)-k}\right ) \Biggr ]} \right \}} \notag \\
    & \leq \exp{ \left \{\omega_{1_{\sigma(\iota^+_{g(t)})}} \tau_D \varkappa \right \}} \sum_{\substack{k \in \mathbb{N}_0; \\ q \leq g(t)}} \exp{ \left\{ - \left ( \omega_{1_{\sigma(\iota^+_{g(t)})}} +\cdots + \omega_{1_{\sigma(\iota^+_{g(t)-q})}}\right ) \tau_D \right \}} \notag \\
    & \exp{\left \{ \substack{ -\Biggl [\beta_{\sigma(\iota^+_{g(t)})} \left (t- \varphi_{m(t)} \right ) + \beta_{\sigma(\iota^+_{g(t)-q_{m(t)-1}})} \left ( \iota_{g(t)-(q_{m(t)-1}-1)}- \varphi_{m(t)-1}\right ) \\ +\cdots + \beta_{\sigma(\iota^+_{g(t)-q})} \left ( \iota_{g(t)-(q-1)} - \varphi_{m(t)-k} \right ) \Biggr ] } \right \}} \label{eq:sum_fsdos_4}
\end{align}
where $\beta_{\sigma(\iota^+_{g(t)-q})} \coloneqq \omega_{1_{\sigma(\iota^+_{g(t)-q})}}- ( \omega_{1_{\sigma(\iota^+_{g(t)-q})}} + \omega_{2_{\sigma (\iota^+_{g(t)-q})}} )/T_*$ and so on. Under the condition of \eqref{eq:FSDoS_condition}, we have $\beta_\sigma > 0$. Equation \eqref{eq:sum_fsdos_4} is bounded from above by
\begin{align}
    & \exp{ \left \{\omega_{1_{\sigma(\iota^+_{g(t)})}} \tau_D \varkappa \right \}} \sum_{\substack{k \in \mathbb{N}_0; \\ q \leq g(t)}} \exp{ \left\{ - \left ( \omega_{1_{\sigma(\iota^+_{g(t)})}} +\cdots + \omega_{1_{\sigma(\iota^+_{g(t)-q})}}\right ) \tau_D \right \}} \notag \\
    & \exp{\left \{ \beta_{\sigma(\iota^+_{g(t)})} \tau_F \eta\right \}} \exp{ \left \{ - \left ( \beta_{\sigma(\iota^+_{g(t)})} + \beta_{\sigma (\iota^+_{g(t)-q_{m(t)-1}})} +\cdots+ \beta_{\sigma(\iota^+_{g(t)-q})}\right ) \tau_F\right \}}, \label{eq:sum_fsdos_5}
\end{align}
using \eqref{eq:assumption_fsdos_1} and so on. Hence equation \eqref{eq:sum_fsdos_1} is bounded from above by
\begin{align}
    & \frac{1}{\underline{\alpha}_{\sigma(\iota^+_{g(t)})}} e^{\bigl ( \omega_{1_{\sigma(\iota^+_{g(t)})}} +\omega_{2_{\sigma(\iota^+_{g(t)})}}\bigr ) \zeta_*} e^{\omega_{1_{\sigma(\iota^+_{g(t)})}} \tau_D \varkappa} e^{\beta_{\sigma(\iota^+_{g(t)})} \tau_F \eta} \notag \\
    & \sum_{\substack{{k \in \mathbb{N}_0;}\\{q \leq g(t)}}} \Biggl [ \left(\nu_{1_{\sigma(\iota^+_{g(t)-q})}} + \nu_{3_{\sigma(\iota^+_{g(t)-q})}} \right ) \prod_{\substack{j \in \mathbb{N}_0; \\ \iota_{g(t)-q} < \iota_{g(t)-j} \leq \iota_{g(t)}}} \left ( \frac{\overline{\alpha}_{\sigma(\iota^+_{g(t)-j})}}{\underline{\alpha}_{\sigma(\iota^+_{g(t)-(j+1)})}}\right ) \notag \\
    & e^{-\sum_{j=0}^q \omega_{1_{\sigma(\iota^+_{g(t)-j})}} \tau_D} e^{- \left ( \beta_{\sigma(\iota^+_{g(t)})} + \beta_{\sigma(\iota^+_{g(t)-q_{m(t)-1}})} + \cdots + \beta_{\sigma (\iota^+_{g(t)-q})}\right )\tau_F } \biggr] \notag \\
    & \leq \overline{\nu} e^{\bigl ( \omega_{1_{\sigma(\iota^+_{g(t)})}} +\omega_{2_{\sigma(\iota^+_{g(t)})}}\bigr ) \zeta_*} e^{\omega_{1_{\sigma(\iota^+_{g(t)})}} \tau_D \varkappa} e^{\beta_{\sigma(\iota^+_{g(t)})} \tau_F \eta} \notag \\
    & \sum_{\substack{k \in \mathbb{N}_0;\\ q \leq g(t)}} \biggl [ \frac{1}{\underline{\alpha}_{\sigma(\iota^+_{g(t)-q})}} \prod_{\substack{j \in \mathbb{N}_0; \\ \iota_{g(t)-q}< \iota_{g(t)-j} \leq \iota_{g(t)}}} \left ( \frac{\overline{\alpha}_{\sigma(\iota^+_{g(t)-j})}}{\underline{\alpha}_{\sigma(\iota^+_{g(t)-j})}}\right ) e^{-\sum_{j=0}^q \omega_{1_{\sigma(\iota_{g(t)-j})}} \tau_D} \notag \\
    & e^{- \left ( \beta_{\sigma(\iota^+_{g(t)})} + \beta_{\sigma(\iota^+_{g(t)-q_{m(t)-1}})} + \cdots + \beta_{\sigma (\iota^+_{g(t)-q})}\right )\tau_F } \biggr] \notag \\
    & \leq \frac{\overline{\nu}}{\underline{\alpha}} e^{\bigl ( \omega_{1_{\sigma(\iota^+_{g(t)})}} +\omega_{2_{\sigma(\iota^+_{g(t)})}}\bigr ) \zeta_*} e^{\omega_{1_{\sigma(\iota^+_{g(t)})}} \tau_D \varkappa} e^{\beta_{\sigma(\iota^+_{g(t)})} \tau_F \eta} \notag \\
    & \sum_{\substack{k \in \mathbb{N}_0;\\ q \coloneqq \sum_{s=0}^k q_{m(t)-s} \leq g(t)}} \biggl [ \prod_{\substack{j \in \mathbb{N}_0; \\ \iota_{g(t)-q}< \iota_{g(t)-j} \leq \iota_{g(t)}}} \left ( \frac{\overline{\alpha}_{\sigma(\iota^+_{g(t)-j})}}{\underline{\alpha}_{\sigma(\iota^+_{g(t)-j})}}\right ) e^{-\sum_{j=0}^q \omega_{1_{\sigma(\iota^+_{g(t)-j})}} \tau_D} \notag \\
    & e^{- \left ( \beta_{\sigma(\iota^+_{g(t)})} + \beta_{\sigma(\iota^+_{g(t)-q_{m(t)-1}})} + \cdots + \beta_{\sigma (\iota^+_{g(t)-q})}\right )\tau_F } \biggr], \label{eq:sum_fsdos_6}
\end{align}
where $\overline{\nu} \coloneqq \max_{\forall \sigma \in \{1, \cdots, n_s\}} \left (\nu_{1_\sigma}+ \nu_{3_\sigma} \right )$. If we can prove the sum term of \eqref{eq:sum_fsdos_6}’s RHS is convergent, then we can say that \eqref{eq:sum_fsdos_1} is bounded from above. We use a ratio test to prove the sum term of \eqref{eq:sum_fsdos_6} is convergent. Consider $qp \coloneqq \sum_{s=0}^p q_{m(t)-s}$ and $qp' \coloneqq \sum_{s=0}^{p+1} q_{m(t)-s}$. The $p$th sum term of that series is
\begin{align}
    a_p \coloneqq \prod_{\substack{j\in \mathbb{N}_0;\\ \iota_{g(t)-qp} < \iota_{g(t)-j} \leq \iota_{g(t)}}} & \left ( \frac{\overline{\alpha}_{\sigma(\iota^+_{g(t)-j})}}{\underline{\alpha}_{\sigma(\iota^+_{g(t)-j})}}\right ) e^{-\sum_{j=0}^{qp} \omega_{1_{\sigma(\iota_{g(t)-j})}} \tau_D} \notag \\
    & e^{- \left ( \beta_{\sigma(\iota^+_{g(t)})} + \beta_{\sigma(\iota^+_{g(t)-q_{m(t)-1}})} + \cdots + \beta_{\sigma (\iota^+_{g(t)-qp})}\right )\tau_F } .
\end{align}
Recall that to compute $a_{p+1}$ all that we need to do is substitute $p + 1$ for all the $p$’s in $a_p$
\begin{align}
    a_{p+1} \coloneqq \prod_{\substack{j\in \mathbb{N}_0;\\ \iota_{g(t)-qp'} < \iota_{g(t)-j} \leq \iota_{g(t)}}} & \left ( \frac{\overline{\alpha}_{\sigma(\iota^+_{g(t)-j})}}{\underline{\alpha}_{\sigma(\iota^+_{g(t)-j})}}\right ) e^{-\sum_{j=0}^{qp'} \omega_{1_{\sigma(\iota_{g(t)-j})}} \tau_D} \notag \\
    & e^{- \left ( \beta_{\sigma(\iota^+_{g(t)})} + \beta_{\sigma(\iota^+_{g(t)-q_{m(t)-1}})} + \cdots + \beta_{\sigma (\iota^+_{g(t)-qp'})}\right )\tau_F } .
\end{align}
Now
\begin{align}
    \frac{a_{p+1}}{a_p}= &\frac{\overline{\alpha}_{\sigma(\iota^+_{g(t)-qp})} \cdots \overline{\alpha}_{\sigma(\iota^+_{g(t)-(qp'-1)})}}{\underline{\alpha}_{\sigma(\iota^+_{g(t)-qp})} \cdots \underline{\alpha}_{\sigma(\iota^+_{g(t)-(qp'-1)})}} e^{-\left (\omega_{1_{\sigma(\iota^+_{g(t)-(qp+1)})}}+ \cdots+ \omega_{1_{\sigma(\iota^+_{g(t)-qp'})}} \right) \tau_D} \notag \\
    & e^{-\beta_{\sigma(\iota^+_{g(t)-qp'})} \tau_F} \notag \\
    =& \frac{\overline{\alpha}_{\sigma(\iota^+_{g(t)-qp})}}{\underline{\alpha}_{\sigma(\iota^+_{g(t)-qp})}} e^{-\omega_{1_{\sigma(\iota^+_{g(t)-(qp+1)})}}\tau_D} \cdots \frac{\overline{\alpha}_{\sigma(\iota^+_{g(t)-(qp'-1)})}}{\underline{\alpha}_{\sigma(\iota^+_{g(t)-(qp'-1)})}} e^{-\omega_{1_{\sigma(\iota^+_{g(t)-qp'})}}\tau_D} \notag \\
    & e^{-\beta_{\sigma(\iota^+_{g(t)-qp'})} \tau_F} .\label{eq:ratio_fsdos}
\end{align}
Now we can say from \eqref{eq:convergance_1}
\begin{align}
     \lim_{p \rightarrow \infty} & \Biggl | \frac{\overline{\alpha}_{\sigma(\iota^+_{g(t)-qp})}}{\underline{\alpha}_{\sigma(\iota^+_{g(t)-qp})}} e^{-\omega_{1_{\sigma(\iota^+_{g(t)-(qp+1)})}}\tau_D} \cdots \frac{\overline{\alpha}_{\sigma(\iota^+_{g(t)-(qp'-1)})}}{\underline{\alpha}_{\sigma(\iota^+_{g(t)-(qp'-1)})}} e^{-\omega_{1_{\sigma(\iota^+_{g(t)-qp'})}}\tau_D} \notag \\
    & e^{-\beta_{\sigma(\iota^+_{g(t)-qp'})} \tau_F} \Biggr| \notag \\
    \leq & \lim_{p \rightarrow \infty} \left | \frac{\overline{\alpha}_{\sigma(\iota^+_{g(t)-qp})}}{\underline{\alpha}_{\sigma(\iota^+_{g(t)-qp})}} e^{-\omega_{1_{\sigma(\iota^+_{g(t)-(qp+1)})}}\tau_D} \right | \cdots \left | \frac{\overline{\alpha}_{\sigma(\iota^+_{g(t)-(qp'-1)})}}{\underline{\alpha}_{\sigma(\iota^+_{g(t)-(qp'-1)})}} e^{-\omega_{1_{\sigma(\iota^+_{g(t)-qp'})}}\tau_D}\right | \notag \\
    & \left |e^{-\beta_{\sigma(\iota^+_{g(t)-qp'})} \tau_F} \right | <1,
\end{align}
because from \eqref{eq:FSDoS_condition}, we can say $\beta_{\sigma(\iota^+_{g(t)-qp'})}>0$ and from Assumption \ref{as2}, $\tau_F> \underline{\Delta}$. If $e^{-\beta_{\sigma(\iota^+_{g(t)-qp'})} \tau_F}<1$ then equation \eqref{eq:FSDoS_condition} and Assumption \ref{as2} are true. Therefore, the upper bound of the sum term of \eqref{eq:sum_fsdos_6} is convergent. So it is bounded from above. Hence, \eqref{eq:sum_fsdos_1} is also bounded from above. Consider the upper bound is $G_2$. \hfill $\blacksquare$ \\
In the same way as in Lemma \ref{lm:upper_bound_sum}, we can prove that
\begin{align}
    & \frac{1}{\underline{\alpha}_{\sigma(\iota^+_{g(t)})}}\mathlarger{\sum}_{\substack{{l \in \mathbb{N};} \\ {0 \leq \iota_{g(t)-l} < \iota_{g(t)}}}} \Biggl \{ \nu_{1_{\sigma(\iota^+_{g(t)-l})}} \mathlarger{\prod}_{\substack{{j \in \mathbb{N}_0;} \\ {\iota_{g(t)-l} < \iota_{g(t)-j} \leq \iota_{g(t)}}}} \left ( \frac{\overline{\alpha}_{\sigma(\iota^+_{g(t)-j})}}{\underline{\alpha}_{\sigma(\iota^+_{g(t)-(j+1)})}} \right ) \notag \\
    & e^{-(\sum_{i=1}^{n_s} \omega_{1_i} |(\mathcal{O}_i \cap \overline{\Upsilon})(\iota_{g(t)-(l-1)},t)|)} e^{ (\sum_{i=1}^{n_s} \omega_{2_i} |(\mathcal{O}_i \cap \overline{\Omega}) (\iota_{g(t)-l},t)|)} \Biggr \}
\end{align}
is bounded from above. Let take the upper bound is $G_3$.
\begin{lemma}\label{lm:convergent_fsdos}
    Under the Assumption \ref{as1}, \ref{as2}, and \ref{as3}, the first term of \eqref{eq:trajectory_fsdos_14}
    \begin{align}
        \prod_{\substack{f \in \mathbb{N}_0; \\ 0 \leq \iota_{g(t)-f} \leq t}} \left (\frac{ \overline{\alpha}_{\sigma(\iota^+_{g(t)-f})}}{ \underline{\alpha}_{\sigma(\iota^+_{g(t)-f})}} \right ) e^{-\sum_{i=1}^{n_s} \omega_{1_i}|(\mathcal{O}_i\cap\overline{\Upsilon})(\iota_0,t)|} e^{\sum_{i=1}^{n_s} \omega_{2_i}|(\mathcal{O}_i\cap \overline{\Omega})(\iota_0,t)|} \label{eq:prod_fsdos_1}
    \end{align}
    is convergent in nature.
\end{lemma}
$Proof$ $of$ $Lemma$ $\ref{lm:convergent_fsdos}:$
Using \eqref{eq:sum_fsdos_2}$-$\eqref{eq:sum_fsdos_6}, upper bound of \eqref{eq:prod_fsdos_1} is
\begin{align}
    & e^{\left (\omega_{1_{\sigma(\iota^+_{g(t)})}} +\omega_{2_{\sigma(\iota^+_{g(t)})}} \right ) \zeta_*} e^{\omega_{1_{\sigma(\iota^+_{g(t)})}} \tau_D \varkappa} e^{\beta_{\sigma(\iota^+_{g(t)})} \tau_F \eta} \prod_{\substack{f \in \mathbb{N}_0; \\ 0 \leq \iota_{g(t)-f} \leq t}} \left (\frac{ \overline{\alpha}_{\sigma(\iota^+_{g(t)-f})}}{ \underline{\alpha}_{\sigma(\iota^+_{g(t)-f})}} \right ) \notag \\
    & e^{-\sum_{j=0}^f \omega_{1_{\sigma(\iota^+_{g(t)-j})}} \tau_D} e^{- \sum_{\substack{k \in \mathbb{N}_0; \\ q \coloneqq \sum_{s=0}^k q_{m(t)-s}\leq f}} \beta_{\sigma(\iota^+_{g(t)-q})} \tau_F}. \label{eq:prod_fsdos_2}
\end{align}
If we can prove that \eqref{eq:prod_fsdos_2} is convergent, then from the direct comparison test, we can say \eqref{eq:prod_fsdos_1} is also convergent. We use a ratio test to prove that \eqref{eq:prod_fsdos_2} is convergent. $p$th term of \eqref{eq:prod_fsdos_2} by avoiding constants are 
\begin{align}
    b_p \coloneqq \prod_{\substack{f \in \mathbb{N}_0; \\ \iota_{g(t)-p} < \iota_{g(t)-f} \leq t}} \left (\frac{ \overline{\alpha}_{\sigma(\iota^+_{g(t)-f})}}{ \underline{\alpha}_{\sigma(\iota^+_{g(t)-f})}} \right ) e^{-\sum_{j=0}^f \omega_{1_{\sigma(\iota^+_{g(t)-j})}} \tau_D} e^{- \sum_{\substack{k \in \mathbb{N}_0; \\ q \leq f}} \beta_{\sigma(\iota^+_{g(t)-q})} \tau_F}.
\end{align}
Hence,
\begin{align}
    b_{p+1} \coloneqq \prod_{\substack{f \in \mathbb{N}_0; \\ \iota_{g(t)-(p+1)} < \iota_{g(t)-f} \leq t}} \left (\frac{ \overline{\alpha}_{\sigma(\iota^+_{g(t)-f})}}{ \underline{\alpha}_{\sigma(\iota^+_{g(t)-f})}} \right ) e^{-\sum_{j=0}^f \omega_{1_{\sigma(\iota^+_{g(t)-j})}} \tau_D} e^{- \sum_{\substack{k \in \mathbb{N}_0; \\ q \leq f}} \beta_{\sigma(\iota^+_{g(t)-q})} \tau_F}.
\end{align}
First, consider FSDoS does not exist in the interval $[\iota_{g(t)-p}, \iota_{g(t)-(p-1)}]$. Therefore,
\begin{align}
    \frac{b_{p+1}}{b_p}= \frac{\overline{\alpha}_{\sigma(\iota^+_{g(t)-p})}}{\underline{\alpha}_{\sigma(\iota^+_{g(t)-p})}} e^{-\omega_{1_{\sigma(\iota^+_{g(t)-p})}} \tau_D}.
\end{align}
Hence \eqref{eq:convergance_1} yields,
\begin{align}
    \lim_{p \rightarrow \infty} \left | \frac{\overline{\alpha}_{\sigma(\iota^+_{g(t)-p})}}{\underline{\alpha}_{\sigma(\iota^+_{g(t)-p})}} e^{-\omega_{1_{\sigma(\iota^+_{g(t)-p})}} \tau_D}\right | <1.
\end{align}
Now consider FSDoS exists in the interval $[\iota_{g(t)-p}, \iota_{g(t)-(p-1)}]$. Therefore,
\begin{align}
    \frac{b_{p+1}}{b_p}= \frac{\overline{\alpha}_{\sigma(\iota^+_{g(t)-p})}}{\underline{\alpha}_{\sigma(\iota^+_{g(t)-p})}} e^{-\omega_{1_{\sigma(\iota^+_{g(t)-p})}} \tau_D} e^{-\beta_{\sigma(\iota^+_{g(t)-p})} \tau_F}.
\end{align}
Hence,
\begin{align}
    & \lim_{p \rightarrow \infty} \left |\frac{\overline{\alpha}_{\sigma(\iota^+_{g(t)-p})}}{\underline{\alpha}_{\sigma(\iota^+_{g(t)-p})}} e^{-\omega_{1_{\sigma(\iota^+_{g(t)-p})}} \tau_D} e^{-\beta_{\sigma(\iota^+_{g(t)-p})} \tau_F} \right | \notag \\
    & \leq \lim_{p \rightarrow \infty} \left \{\left | \frac{\overline{\alpha}_{\sigma(\iota^+_{g(t)-p})}}{\underline{\alpha}_{\sigma(\iota^+_{g(t)-p})}} e^{-\omega_{1_{\sigma(\iota^+_{g(t)-p})}} \tau_D}\right | \left | e^{-\beta_{\sigma(\iota^+_{g(t)-p})} \tau_F} \right | \right \} <1
\end{align}
$\forall \beta_{\sigma} \in \mathbb{R}_{>0}, \tau_F \in \mathbb{R}_{\geq \underline{\Delta}},$ which concludes that \eqref{eq:prod_fsdos_1} is convergent in nature. \hfill $\blacksquare$\\
Hence from \eqref{eq:trajectory_fsdos_14}, we readily get
\begin{align}
    \|x(t)\| \leq & \sqrt{\frac{\overline{\alpha}_{\sigma(\iota^+_0)} \cdots \overline{\alpha}_{\sigma(\iota^+_{g(t)})}}{\underline{\alpha}_{\sigma(\iota^+_0)} \cdots \underline{\alpha}_{\sigma(\iota^+_{g(t)})}}} e^{-\frac{1}{2} \sum_{i=1}^{n_s} \omega_{1_i}| ( \mathcal{O}_i \cap \overline{\Upsilon} )(\iota_0,t)|} e^{\frac{1}{2} \sum_{i=1}^{n_s} \omega_{2_i} |(\mathcal{O}_i \cap \overline{\Omega})(\iota_0,t)|} \|x(\iota_0)\| \notag \\
    & + \sqrt{G_2 +G_3} \|w_t\|_\infty, \label{eq:trajectory_fsdos_final}
\end{align}
$\forall t \in [0,\varphi_{m(t)}+v_{m(t)}).$ Hence, \eqref{eq:trajectory_fsdos_final} is ISS. \hfill $\blacksquare$

 \bibliographystyle{elsarticle-num} 
 \bibliography{Ref.bib}





\end{document}